\documentclass[aps,prx,reprint,superscriptaddress,floatfix]{revtex4-2}

\usepackage{amsmath,amssymb,amsfonts}
\usepackage{amsthm}
\usepackage{graphicx}
\usepackage{booktabs}
\usepackage{xcolor}
\usepackage[hidelinks]{hyperref}
\usepackage{braket}
\usepackage{bbm}
\usepackage{multirow}
\usepackage{placeins}

\renewcommand{\ket}[1]{|#1\rangle}

\newcommand{\op}[1]{\hat{#1}}
\newcommand{\Z}{\mathbb{Z}}

\newcommand{\nred}{N_{\mathrm{red}}}
\newcommand{\nso}{N_{\mathrm{so}}}
\newcommand{\GFtwo}{\mathbb{F}_2}
\definecolor{spinparity}{HTML}{1B8A7A}
\definecolor{halftranslation}{HTML}{D52E43}
\definecolor{pointgroup}{HTML}{7167A8}
\newcommand{\legendsquare}[1]{\textcolor{#1}{\rule{1.4ex}{1.4ex}}}
\newtheorem{proposition}{Proposition}

\begin{document}

\title{Periodic Symmetry-Adapted Encoding: Qubit Reduction in Crystalline Electronic Structure}

\author{Dario Picozzi}
\email{picozzi.dario@gmail.com}
\affiliation{Department of Physics and Astronomy, University College London (UCL),
Gower Street, London, WC1E 6BT, United Kingdom}
\affiliation{London Centre for Nanotechnology, 19 Gordon St, London, WC1H 0AH,
United Kingdom}

\begin{abstract}
We introduce periodic symmetry-adapted encoding (SAE) for
qubit-efficient simulations of crystalline materials, extending the molecular
SAE framework developed in earlier work to periodic systems.  The method constructs a
$\Gamma$-point supercell active-space Hamiltonian from $k$-point Hartree--Fock
data and identifies independent Boolean symmetry generators arising from spin
parity, crystal translations, and point-group operations.  Each generator
removes one qubit while preserving the Hamiltonian spectrum in the target
symmetry sector.  The inclusion of crystal translations increases the maximum
number of independent generators, and hence the number of removable qubits,
from five in molecular SAE to eight in periodic SAE.  Across ten crystals
spanning cubic, hexagonal, trigonal, and tetragonal structures, the method
removes four to eight qubits.  In CsCl, all eight generators are realised,
reducing CAS(6,7) from 14 qubits to 6.  Complete-spectrum comparisons verify
target-sector isospectrality, while noiseless UCCSD-VQE reaches the
fixed-particle FCI energies throughout the suite.  The resulting circuits
reduce unoptimised logical CNOT counts by up to 99.7\% relative to the full
Jordan--Wigner encoding.
The method is implemented in the open-source \texttt{QuantumSymmetry} package.
\end{abstract}

\maketitle

% ============================================================
\section{Introduction}
\label{sec:intro}
% ============================================================

Simulating the electronic structure of crystalline materials is a major target
application for quantum computers~\cite{Babbush2018,Su2021,Bauer2020,
YoshiokaPeriodic2022,Rubin2023,Ma2020,Liu2022}.
Periodic systems pose distinct challenges already familiar from solid-state
electronic structure~\cite{AshcroftMermin1976,Martin2004}: the Hamiltonian is
naturally expressed in a Bloch basis at multiple $k$-points, the active space is
selected from crystal bands with their $k$-point provenance, and its spatial
symmetries form the crystal space group~\cite{BradleyCracknell2010}.
Existing quantum simulation proposals for periodic systems have largely focused on
plane-wave bases~\cite{Babbush2018,Su2021}, second-quantised Hubbard models, or
direct $k$-point encodings~\cite{YoshiokaPeriodic2022}.
Gaussian orbital and active-space routes to correlated solids provide a
chemistry-like starting point for quantum algorithms
~\cite{McClain2017,LiuK2G2020,Ma2020}.

The symmetry-adapted encoding (SAE) framework~\cite{Picozzi2023,PicozziTennysonCASBK} addresses
the molecular version of this problem.
SAE identifies physical $\mathbb{Z}_2$ spin-parity and point-group symmetries of
a second-quantised Hamiltonian, encodes them as a system of Boolean
linear equations in the spin-orbital occupancies, and removes redundant qubits via
affine Clifford transformations, projecting into the physical symmetry sector.
Combined with the complete active space (CAS) approximation, SAE-CAS reduces
molecular active-space qubit counts without changing the spectrum in the
selected symmetry sector.
General Pauli-symmetry tapering can recover an equivalent Hamiltonian block and,
when its Clifford transformation is retained, can also transform other mapped
operators~\cite{Bravyi2017}.  Periodic SAE instead derives its generators
directly from spin and crystal-orbital characters, without a separate
Pauli-commutant search.  A single affine map applies to the Hamiltonian,
fermionic observables, and excitation generators while retaining selected
spin-orbital occupations as readable reduced coordinates.  This physics-labelled
construction also differs from point-group reductions targeted specifically at
finite molecules~\cite{Setia2020}.  In the periodic setting the spatial
characters come from finite crystal translations and space-group operations
after their actions have been restricted to the selected active space.
Molecular SAE can realise at most five independent Boolean generators, whereas
the finite translation group can supply up to three additional independent
Boolean characters in three dimensions.

Previous work developed momentum-conserving periodic UCC constructions and
demonstrated symmetry tapering for a periodic hydrogen-chain
model~\cite{Manrique2021}.  Here we extend the symmetry-adapted encoding
framework to crystalline active spaces, jointly converting spin,
crystal-translation, and point symmetries into reductions in qubit count
and ansatz complexity.
Recent periodic quantum-simulation workflows reduce variational and circuit
costs through adaptive operator pools and variance extrapolation
~\cite{LiFanLiu2025}, or through Wannier localisation, hybrid fermion-to-qubit
encoding, swap-network compilation, and measurement optimisation
~\cite{Clinton2024}.

The present work makes two distinct contributions.
First, before any fermion-to-qubit mapping, we construct the fermionic
active-space Hamiltonian by transforming primitive-cell $k$-point
density-fitted quantities directly into a selected real $\Gamma$-supercell
active block.  This avoids materialising a full supercell atomic-orbital
four-index tensor while retaining the decomposition over primitive-cell
$k$ components and the symmetry representation of the active orbitals.  Second,
periodic SAE constructs spin, crystal-translation, and point-symmetry generators
directly from the retained orbital-character information and uses them to reduce
the fermionic Hamiltonian, without a separate Pauli-commutant search.  The
resulting affine reduction projects the selected symmetry sector and provides a
common within-sector mapping for the Hamiltonian, fermionic observables,
reference determinant, and excitation operators.  Both contributions are
compatible with the adaptive-ansatz and circuit-compilation strategies above.

We benchmark the method on ten insulating or semiconducting crystals spanning
cubic, hexagonal, trigonal, and tetragonal structures, all represented as
$(2,2,2)$ folded supercells.  Their degeneracy-complete active spaces contain
2--6 active electrons in 4--8 active orbitals, corresponding to 8--16
Jordan--Wigner (JW) spin-orbital qubits.  Periodic SAE removes 4--8 qubits across
the set; CsCl gives the largest
reduction, from 14 to 6 qubits.  For every system we run exact diagonalisation
and a variational quantum eigensolver with a unitary coupled-cluster
singles-and-doubles ansatz (UCCSD-VQE); full JW-UCCSD-VQE is run for nine systems.

% ============================================================
\section{Method}
\label{sec:method}
% ============================================================

The calculation has two stages.  First, a primitive-cell $k$-point restricted
Hartree--Fock (KRHF) calculation is folded to a commensurate $\Gamma$-point
supercell and a complete frontier active manifold is specified.  Finite
translations and space-group operations are represented on this manifold;
a maximal independent set whose restricted actions are commuting real
involutions then defines a simultaneous-character orbital
basis.  The primitive-cell density-fitted quantities are contracted directly
into that symmetry-adapted active block, and numerical symmetry restoration
gives $\op{H}_{\mathrm{sym}}$.  Second, the orbital characters define the matrix
$A$, the reference determinant fixes the target vector $c$, and the affine
Clifford transformation restricts $\op{H}_{\mathrm{sym}}$ to that sector.

\subsection{Periodic active-space Hamiltonian}
\label{subsec:periodic_ham}

The periodic electronic Hamiltonian in a Gaussian-type orbital (GTO) basis is
naturally expressed in terms of Bloch orbitals sampled on an $N_k$-point mesh.
A KRHF calculation yields Bloch orbitals $\phi_{n\mathbf k}$
with coefficients $C^{\mathbf k}_{\mu n}$ in the $k$-adapted atomic-orbital
(AO) basis.  For a complete diagonal mesh $(N_0,N_1,N_2)$,
$N_k=N_0N_1N_2$, and the commensurate supercell contains $N_k$
primitive-cell images.  The same finite one-electron space can be represented
at the $\Gamma$ point of that supercell~\cite{LiuK2G2020,PySCF2020,PySCF2018}.
Let $\mathbf R$ label the primitive-cell translations in the supercell and let
$m$ label the full set of supercell orbitals before active-space selection.  The
change of orbital basis can be written
\begin{equation}
\begin{aligned}
  C^{\mathrm{SC}}_{\mathbf R\mu,m}
  &=\frac{1}{\sqrt{N_k}}\sum_{\mathbf k,n}
    e^{i\mathbf k\cdot\mathbf R}C^{\mathbf k}_{\mu n}W_{\mathbf k n,m},\\
  W^\dagger W&=I.
\end{aligned}
\label{eq:k2gamma-fold}
\end{equation}
The phase factor is the normalised discrete Fourier transformation between the
complete $k$-point mesh and the supercell translations, and $W$ is a unitary
change of basis within the same orbital space.  For the spin-free,
time-reversal-symmetric calculations considered here, $W$ can be chosen to
produce real supercell orbitals; subsequent symmetry adaptation applies real
orthogonal rotations within near-degenerate blocks.

No orbitals are discarded by these transformations.  The final real supercell
orbitals are orthonormal and span the same finite one-electron space as the
original Bloch orbitals, although an individual final supercell orbital need
not have a definite primitive-cell $\mathbf k$ label.  Active-space selection and crystal-symmetry
analysis are performed in this real supercell-orbital basis.

For each benchmark we choose a frontier active space containing the highest
$n_{\mathrm{occ}}$ occupied and lowest $n_{\mathrm{unocc}}$ unoccupied supercell
orbitals.  This gives CAS$(n_e,n_{\mathrm{act}})$, with
$n_e=2n_{\mathrm{occ}}$ active electrons and
$n_{\mathrm{act}}=n_{\mathrm{occ}}+n_{\mathrm{unocc}}$ active spatial orbitals.
The contribution from the inactive doubly occupied orbitals, referred to below
as the frozen-core space, is contracted exactly within the chosen active-space
approximation, yielding an effective active-space one-body operator (including
the core mean field) and the two-electron repulsion integrals restricted to the
active block~\cite{Helgaker2000}.
The active-space second-quantised Hamiltonian is then
\begin{equation}
  \op{H} = E_{\mathrm{core}} + \sum_{pq} h_{pq}^{\mathrm{eff}}\, a_p^\dagger a_q
  + \frac{1}{2}\sum_{pqrs} g_{pqrs}\, a_p^\dagger a_q^\dagger a_s a_r,
  \label{eq:ham}
\end{equation}
where $p,q,r,s$ index the $\nso = 2n_{\mathrm{act}}$ active spin-orbitals in
interleaved spin order, $E_{\mathrm{core}}$ is the frozen-core energy (including
nuclear repulsion), $h_{pq}^{\mathrm{eff}}$ is the core-dressed one-electron
integral, and $g_{pqrs}$ are the two-electron repulsion integrals.
Before the numerical symmetry restoration described below, the finite-precision
operator obtained at this stage is denoted $\op{H}_{\mathrm{num}}$.

The calculations use the density-fitting machinery of
\textsc{PySCF}~\cite{PySCF2020,PySCF2018}: primitive-cell Gaussian density-fitting
(GDF) quantities~\cite{Sun2017} are contracted directly into the supercell
active block, including the frozen-core and finite-supercell Ewald exchange
contributions~\cite{Ewald1921,Fraser1996}.  Appendix~\ref{app:folded-integrals}
gives the contractions and Ewald sign convention, while
Appendix~\ref{app:reduction-validation} reports an independent check of the
GDF assembly.

This direct active-block construction retains the primitive-cell $k$
decomposition and orbital information needed for the symmetry analysis below.

\subsection{Periodic symmetry-adapted encoding}
\label{subsec:periodic-sae}

Periodic SAE represents physical symmetries of the crystalline active space as
Boolean constraints on spin-orbital occupations.  The crystal operations are
represented on the folded active orbitals, a simultaneous-character basis is
constructed for the retained commuting involutions, and the physical symmetry
sector is fixed by an affine Clifford transformation.
We first formulate the reduction for a given character matrix $A$, and then
describe how the periodic generators and simultaneous-character basis used to
construct $A$ are obtained.

\subsubsection{Jordan--Wigner encoding and spin-orbital ordering}
\label{subsec:jw}

The active-space Hamiltonian~\eqref{eq:ham} is mapped to a qubit operator via the
Jordan--Wigner (JW) transformation~\cite{Jordan1928}:
\begin{equation}
  a_p^\dagger \mapsto
  \frac{X_p - iY_p}{2} \prod_{j<p} Z_j,
  \qquad
  a_p \mapsto
  \frac{X_p + iY_p}{2} \prod_{j<p} Z_j,
\end{equation}
where $X_p$, $Y_p$, $Z_p$ are Pauli operators acting on qubit $p$.
Spin-orbitals are ordered in interleaved spin order: qubits $0,2,4,\ldots$ carry
spin-up electrons and qubits $1,3,5,\ldots$ carry spin-down electrons.
This ordering makes the spin-parity symmetries diagonal in the $Z$-basis and
facilitates identification of the Boolean generators described below.

The JW Hamiltonian acts on $\nso$ qubits.  For example, the CAS(6,8) MgF$_2$
active space has $\nso=16$ and a $2^{16}=65536$ dimensional full Hilbert space,
whereas its physical $(N_\uparrow,N_\downarrow)=(3,3)$ sector has dimension
$\binom{8}{3}^2=3136$.

\subsubsection{Boolean constraints and affine reduction}
\label{subsec:generators}

The SAE framework~\cite{Picozzi2023,PicozziTennysonCASBK} identifies $\mathbb{Z}_2$
symmetries of $\op{H}$ of the form
\begin{equation}
  \tau_j = \bigotimes_{p=0}^{\nso-1} Z_p^{A_{jp}},
  \qquad A_{jp} \in \{0,1\},
  \label{eq:generator}
\end{equation}
where $A_{jp}=1$ if and only if symmetry generator $j$ has character $-1$ on
spin-orbital $p$.
These mutually commuting operators satisfy $[\tau_j, \op{H}] = 0$ and
$\tau_j^2 = \mathbbm{1}$, so $\op{H}$ and the $\tau_j$ admit a common
eigenbasis in which every $\tau_j$ has eigenvalue $\pm1$.
For $n_g$ independent generators, the generated Boolean symmetry group is
isomorphic to $\Z_2^{n_g}$.
Fixing the physical eigenvalue of each generator removes one qubit per generator
via an affine Clifford transformation, projecting the Hamiltonian into the
target symmetry sector without changing its spectrum in that sector.

The rows of $A$ are constructed directly from the orbital characters of
physical symmetries, without inspecting the Pauli terms of $\op{H}$.
For an active-space determinant with Jordan--Wigner occupation vector
$a\in\GFtwo^{\nso}$, the generator eigenvalue is
\begin{equation}
  \tau_j\ket{a}=(-1)^{(Aa)_j}\ket{a}.
  \label{eq:character-constraint}
\end{equation}
Choosing a target sector $c\in\GFtwo^{n_g}$ therefore imposes the Boolean
constraints
\begin{equation}
  A a = c .
  \label{eq:sector-constraints}
\end{equation}
For the benchmarks below, $a_{\mathrm{ref}}$ is the occupation vector of the
folded KRHF reference determinant and the target is selected operationally as
\begin{equation}
  c=Aa_{\mathrm{ref}}\pmod 2,
  \qquad s_j=(-1)^{c_j}.
  \label{eq:reference-sector}
\end{equation}
Any row operations used to obtain independent pivot constraints are applied to
the augmented matrix $[A\,|\,c]$, so they preserve the same physical sector.
These equations are the periodic analogue of the point-group character
constraints in molecular SAE~\cite{Picozzi2023,PicozziTennysonCASBK}.
The number of removable qubits is the Boolean rank
$n_g=\operatorname{rank}_{\GFtwo}A$.  Within the spin-parity and spatial-
symmetry construction used here, molecular SAE satisfies $n_g\leq5$, whereas
a supercell generated by a diagonal $(N_0,N_1,N_2)$ mesh satisfies
$n_g\leq5+e\leq8$, where $e$ is the
number of even folded-mesh axes.  These bounds and their saturation by the
CsCl CAS(6,7) active space are proved in Appendix~\ref{app:boolean-rank}.
After row reduction over $\GFtwo$, pivot spin-orbital occupations are redundant
and the non-pivot occupations are kept as reduced coordinates.
After reordering the occupation vector so that the $n_g$ pivots come first, let
$T$ be the identity with its first $n_g$ rows replaced by the independent rows
of $A$, and set $b=(c^{\mathsf T},0,\ldots,0)^{\mathsf T}$.  The basis change is
the affine map
\begin{equation}
  \ket{a} \longmapsto \ket{q}=\ket{T a \oplus b}.
  \label{eq:periodic-affine-map}
\end{equation}
Its Clifford tableau is the affine SAE tableau~\cite{Picozzi2023,
PicozziTennysonCASBK}; after projection of the first $n_g$ symmetry bits, the
Hamiltonian acts on $\nred=\nso-n_g$ qubits.
Because the remaining bits are the chosen non-redundant spin-orbital occupation
coordinates, every reduced computational-basis state corresponds to a unique
Slater determinant in the sector~\eqref{eq:sector-constraints}.
Fermionic observables and excitation generators are mapped through the same
Jordan--Wigner, affine-Clifford, and sector-projection steps.  Let $P_c$ denote
the projector onto Eq.~\eqref{eq:sector-constraints}.  Each monomial in a single
(double) excitation flips two (four) occupation bits.  It preserves the
eigenvalue of $\tau_j$ exactly when an even number of the flipped bits satisfy
$A_{jp}=1$; an odd number reverses it.

The anti-Hermitian generator $G_\mu$ associated with one packed spin-singlet
UCCSD amplitude sums the allowed spin assignments involving the same spatial-
orbital indices and their Hermitian conjugates.  The selected spatial
symmetries act identically on the two spin partners of each spatial orbital,
while every term preserves both spin parities.  This packing therefore
preserves Boolean character: every monomial in $G_\mu$ has the same character
under every selected symmetry, and $G_\mu$ itself has a definite character.
The screening rule is simply
\begin{equation}
  P_cG_\mu P_c=0
  \quad\text{if}\quad
  \tau_jG_\mu\tau_j=-G_\mu\ \text{for any }j.
  \label{eq:excitation-screening}
\end{equation}
Such a generator maps a target-sector state into a different symmetry sector,
so its amplitude is absent from the SAE circuit.  A generator that commutes
with every $\tau_j$ preserves the sector, and restricting its exponential to
that sector is equivalent to exponentiating its restricted generator.  The
retained UCCSD generators therefore have an exact constructive representation
on the reduced register.

\subsubsection{Crystal-symmetry generators}

For periodic systems we identify three classes of generators:

\paragraph{Spin-parity symmetries.}
For $\sigma\in\{\uparrow,\downarrow\}$, the operator
$P_\sigma=\prod_{m=0}^{n_{\mathrm{act}}-1}(-1)^{\hat n_{m\sigma}}$ measures the
parity of the number of spin-$\sigma$ electrons in the active space and is an
exact symmetry of~\eqref{eq:ham}.
These contribute two generators, present in every system.

\paragraph{Point-group symmetries.}
Point operations realised by primitive-cell space-group symmetries---including
rotations, reflections, and inversion---are represented as matrices on the
supercell active-orbital
space~\cite{Picozzi2023,PicozziTennysonCASBK}.  An operation
supplies a $\mathbb Z_2$ generator when its restriction to the active space is a
real involution and the symmetry-adaptation step gives every active orbital a
definite eigenvalue $\pm1$.  The operation need not have order two outside the
active space.
Encoding generators are obtained by lifting primitive-cell space-group
operations to the supercell AO basis; finite-supercell point-group labels are
diagnostic only.

\paragraph{Crystal translation symmetries.}
The finite Born--von K\'arm\'an translation group contains the shifts
$T_{\sum_i m_i\mathbf a_i}$ with $0\leq m_i<N_i$.  In the supercell AO basis,
each shift permutes equivalent atomic orbitals.  We enumerate these finite
translations and restrict their actions to the retained active orbitals.  A
translation supplies a $\mathbb Z_2$ generator when this restricted action
squares to $+I$ and is non-trivial; its $\pm1$ orbital signs then define the
Boolean character.  This test determines the useful translation from the
active representation rather than assuming it in advance.

The familiar half-supercell shift $T_{\mathbf A_i/2}$ is always an order-two
candidate when $N_i$ is even, but it can act as $+I$ on a particular active
space and add no constraint.  In that case a shorter translation can still
become an involution after restriction if its square acts as $+I$ on that
space.  More generally, all finite translations are tested, so coupled
directions and non-coordinate representatives are included automatically.
For a diagonal $(N_0,N_1,N_2)$ mesh, the resulting translation-derived Boolean
rank is at most the number of even mesh axes and therefore at most three.
Together with the two spin-parity generators and up to three independent
non-translation generators, these translations make an eight-generator
$\Z_2^8$ reduction possible.  The $(2,2,2)$ mesh used in the benchmarks is the
minimal three-dimensional mesh exposing all three half-supercell translations.
Figure~\ref{fig:supercell_geom} shows the full set of eight generators for the
CsCl benchmark system.

\begin{figure*}[t]
  \centering
  \includegraphics[width=0.88\textwidth]{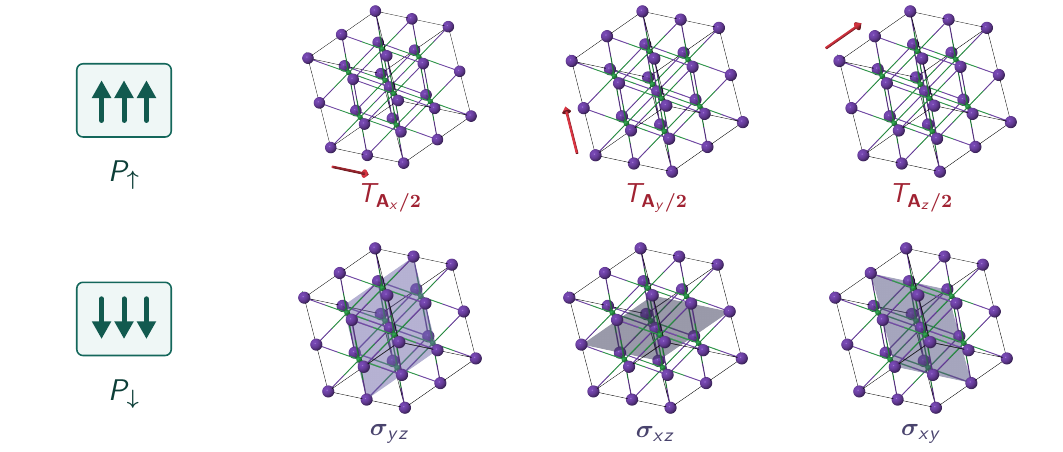}
  \caption[Caesium chloride (CsCl) symmetry generators]{%
    Eight commuting symmetry generators for the B2 caesium chloride (CsCl)
    CAS(6,7) active space in the $(2,2,2)$ simple-cubic supercell.
    The generators are the spin-up and spin-down electron-number parities
    $P_\uparrow$ and $P_\downarrow$, the three half-supercell translations
    $T_{\mathbf A_x/2}$, $T_{\mathbf A_y/2}$, and $T_{\mathbf A_z/2}$, and the
    three reflections $\sigma_{yz}$, $\sigma_{xz}$, and $\sigma_{xy}$.  The
    first column shows the spin parities; each remaining column pairs a
    translation with the reflection normal to the same axis.  Cs atoms are
    purple and Cl atoms green.  The resulting $\Z_2^8$ group reduces the
    register exactly from 14 to 6 qubits.
  }
  \label{fig:supercell_geom}
\end{figure*}

\subsubsection{Selecting a commuting generator set}

Every finite supercell translation and every translated representative of the
primitive-cell space group is lifted to the supercell atomic-orbital basis.  Let
$C_{\mathcal A}$ be the matrix of active-orbital coefficients and let $S$ be the
AO overlap matrix.
For a candidate operation $U_g$, its active action is
\begin{equation}
  M_g=C_{\mathcal A}^{\dagger}S U_g C_{\mathcal A}.
  \label{eq:active-symmetry-restriction}
\end{equation}
We first require the active, frozen-core, and active occupied subspaces to be
invariant.  We then retain the candidate precisely when the raw restricted
matrix is a real involution, $M_g^2=I$.  Operations are compatible when their
restricted matrices commute, $[M_g,M_h]=0$; they need not commute or square to
the identity on discarded orbitals.  Compatible candidates are simultaneously
diagonalised within active blocks of equal occupation and approximately equal
energy.  Finally, row reduction of their $\pm1$ characters, together with the
two spin parities, selects the set with the largest active-space Boolean rank.

The retained active-space involutions have exact eigenvalues $\pm1$; deviations
in their finite AO/block representations are numerical.  For each retained
operation, every numerical eigenvalue $\lambda_p$ is therefore replaced by its
sign, defining the exact diagonal representative
$\widetilde U_g=\operatorname{diag}[\operatorname{sgn}(\lambda_p)]$.  The sign
on active spin orbital $p$ is encoded as $A_{jp}=0$ for $+1$ and $A_{jp}=1$ for
$-1$.  Row reduction selects the independent constraints used in
Eq.~\eqref{eq:sector-constraints}.  The Boolean constraints and group average
below use these sign-restored representatives rather than the numerical block
matrices.

\subsubsection{Numerical symmetry restoration}

Finite density-fitting and integral-transformation tolerances can leave small
matrix elements that violate these otherwise exact crystal symmetries.  After the
generator set is selected, we restore the selected spatial Boolean group
$G_{\mathrm{sp}}$ by the group average, writing $U_g$ for its exact sign-restored
representatives,
\begin{equation}
  \op{H}_{\mathrm{sym}}
  = \frac{1}{|G_{\mathrm{sp}}|}\sum_{g\in G_{\mathrm{sp}}}
    \op{U}_g\op{H}_{\mathrm{num}}\op{U}_g^\dagger .
  \label{eq:symmetry-restoration}
\end{equation}
In the simultaneous $\pm1$-character orbital basis this is equivalent to setting
symmetry-forbidden one- and two-electron tensor elements to zero.  The generator
rank and qubit reduction are fixed before this numerical restoration, which
enforces the spatial symmetries identified from the input crystal and represented
on the selected active manifold.  Only the retained spatial generators enter
this average; the spin parities are exact by construction because
$\op{H}_{\mathrm{num}}$ conserves $N_\uparrow$ and $N_\downarrow$.
As proved in Appendix~\ref{app:symmetry-restoration}, the average is the
Hilbert--Schmidt orthogonal projection onto the symmetry-invariant operator
algebra and obeys
$P_c\op{H}_{\mathrm{sym}}P_c=P_c\op{H}_{\mathrm{num}}P_c$ in every spatial-
character sector.  It therefore removes only couplings between these sectors
and leaves each compressed Hamiltonian block exactly unchanged; restriction to
the fixed spin sector is unaffected by the already exact spin conservation.
The affine SAE transformation is then exactly isospectral within each selected
symmetry sector of $\op{H}_{\mathrm{sym}}$.

\subsection{Benchmark protocol and computational details}
\label{subsec:computational-details}

Appendix~\ref{app:benchmark-inputs} records the cell parameters, basis and
pseudopotential choices, and folded-supercell active-orbital windows.  Each of the
ten benchmark calculations starts from a closed-shell
\textsc{PySCF} KRHF calculation on the primitive cell with a $(2,2,2)$
$k$-point mesh and \texttt{exxdiv=ewald}.

Applicability beyond self-inverse $k$ meshes is checked end to end for a
$(4,2,2)$ CsCl mesh in Appendix~\ref{app:reduction-validation}.

After the symmetry generator set and simultaneous-character basis have been
chosen, the primitive-cell GDF quantities are folded directly into the
resulting real $\Gamma$-point supercell active block.  The default
basis/pseudopotential pair is
\texttt{gth-szv}/\texttt{gth-pade}; NaCl and CsCl use the exceptions listed in
Appendix~\ref{app:benchmark-inputs}.  The symmetry adaptation uses an
energy-block tolerance of $5\times10^{-3}$~Ha.  The root-mean-square active- and
frozen-core-subspace leakage and the normalised involution residual are each
required to be below $10^{-5}$; the active occupied-subspace leakage is below
$10^{-4}$.  During simultaneous diagonalisation, retained operations have an
inter-block Frobenius norm below $10^{-2}$ and block eigenvalues within
$10^{-2}$ of $\pm1$; pairwise active-space commutator Frobenius norms are at
most $10^{-6}$.
The active windows are separately checked for degeneracy closure at $10^{-4}$,
$10^{-3}$, and $5\times10^{-3}$~Ha.

The reported JW and SAE operators use the symmetry-restored active-space
Hamiltonian.  Exact physical references are obtained by restricting the
unreduced Hamiltonian to the fixed $(N_\uparrow,N_\downarrow)$ sector.  The ten
resource benchmarks start from the folded KRHF determinant and use spin-singlet
fermionic UCCSD generators.

Resource circuits are singlet UCCSD circuits~\cite{Romero2018} with one indexed
anti-Hermitian generator $G_\mu$ per packed UCCSD amplitude.  We form three
controlled comparisons.  Full JW uses every such generator; at the fermionic
generator level, this enumeration conserves particle number and spin projection
but does not screen by the selected spatial characters.  The symmetry-filtered
comparator JW-SF retains the full JW Hamiltonian, register, reference state, and
mapping, but includes only generators whose SAE sector image is nonzero.  SAE applies the same
filtering automatically through Eq.~\eqref{eq:excitation-screening} and
represents the identical indexed subset on the reduced register.  Thus JW
versus JW-SF isolates ansatz filtering, whereas JW-SF versus SAE
isolates the effect of reduced-register encoding.  Figure~\ref{fig:resource_summary}
shows the full-JW and SAE endpoints; the complete three-way comparison is given
in Table~\ref{tab:jw-sf-resources}.

The reported depth and CNOT metrics are unoptimised logical circuit counts after
three successive Qiskit decompositions.  Qiskit 2.4.1 uses its default
first-order, one-repetition Lie--Trotter synthesis with a chain CNOT structure;
no hardware connectivity or transpiler optimisation is assumed.  The JW and
SAE UCCSD circuits start from the target-number folded KRHF determinant and
conserve $N_\uparrow$ and $N_\downarrow$, so their ideal VQE states remain in the
fixed-particle sector.  Representative level-3 compilations for diamond, CsCl,
and $\alpha$-quartz use Qiskit 2.4.1 with transpiler seed 271828 and the native
basis \texttt{rz}, \texttt{sx}, \texttt{x}, and \texttt{cx}, under both
all-to-all and bidirectional nearest-neighbour-line connectivity.

For the ten resource benchmarks, noiseless VQE energies are evaluated from
exact statevectors~\cite{McClean2016}.  The Hamiltonians, ansatz circuits,
sequential least-squares programming (SLSQP) optimiser, and convergence settings
are common across that comparison.
Finite-difference SLSQP starts from zero
UCCSD amplitudes and uses at most 300 optimiser iterations with
\texttt{ftol}=$10^{-9}$; the reported optimisation count is the number of
objective-function evaluations.  JW-UCCSD-VQE is not run for MgF$_2$ because of
the substantially greater cost of its 16-qubit, 135-parameter statevector
optimisation.  Exact fixed-$(N_\uparrow,N_\downarrow)$ diagonalisation is run for
sector dimension at most 10{,}000 and SAE VQE for at most 11 reduced qubits; all
ten benchmarks satisfy these limits.

The Hamiltonian construction, circuit-resource calculations, and spectrum
validation use \texttt{QuantumSymmetry} v0.3.4, PySCF 2.13.0, OpenFermion
1.7.1, and Qiskit 2.4.1.  Pauli counts include the identity and refer to the
active-space Hamiltonian.  Fermionic
coefficients below $10^{-14}$~Ha are omitted during operator assembly, with no
additional Pauli-magnitude cutoff.  Tensor elements, Pauli-term counts, and
coefficient one-norms refer to the chosen symmetry-adapted orbital basis; residual
unitary freedom within degenerate blocks can
alter term counts by a few percent without affecting the spectrum or any
reported energy.

% ============================================================
\section{Results}
\label{sec:results}
% ============================================================

\subsection{Benchmark systems}
\label{subsec:systems}

We benchmark ten insulating or semiconducting crystals in $(2,2,2)$ folded
supercells: diamond, silicon, 3C-SiC, MgO, NaCl, CsCl, h-BN, AlN,
$\alpha$-quartz, and MgF$_2$.  The suite spans cubic, hexagonal, trigonal, and
tetragonal structures; Fig.~\ref{fig:benchmark_supercells} gives each structure,
space group, and active space.  Gaussian-type orbital basis sets and
Goedecker--Teter--Hutter pseudopotentials~\cite{GTH1996,Hartwigsen1998} are used
unless noted otherwise.

The active spaces are selected to be complete with respect to the relevant
near-degenerate manifolds at the sampled valence- and conduction-band edges,
avoiding active windows that cut through a degenerate band edge.
This gives CAS$(6,6)$ for diamond, 3C-SiC, and h-BN; CAS$(6,7)$ for silicon,
MgO, NaCl, and CsCl; CAS$(2,8)$ for AlN; CAS$(4,4)$ for $\alpha$-quartz; and
CAS$(6,8)$ for MgF$_2$.
The corresponding JW registers contain 8--16 spin-orbital qubits.
The CsCl active space contains seven spatial orbitals and realises the full
$\Z_2^8$ Boolean symmetry group available on the $(2,2,2)$ mesh: two spin
parities, three half translations, and three point-group reflections.
Figure~\ref{fig:cscl_active_orbitals} shows the seven symmetry-adapted active
orbitals, labelled by their space-group irreps, together with the spin-orbital
occupations and the eight symmetry-fixed coordinates removed by SAE.
Figure~\ref{fig:benchmark_supercells} shows all ten benchmark supercells.

\begin{figure*}[t]
  \centering
  \includegraphics[width=0.72\textwidth]{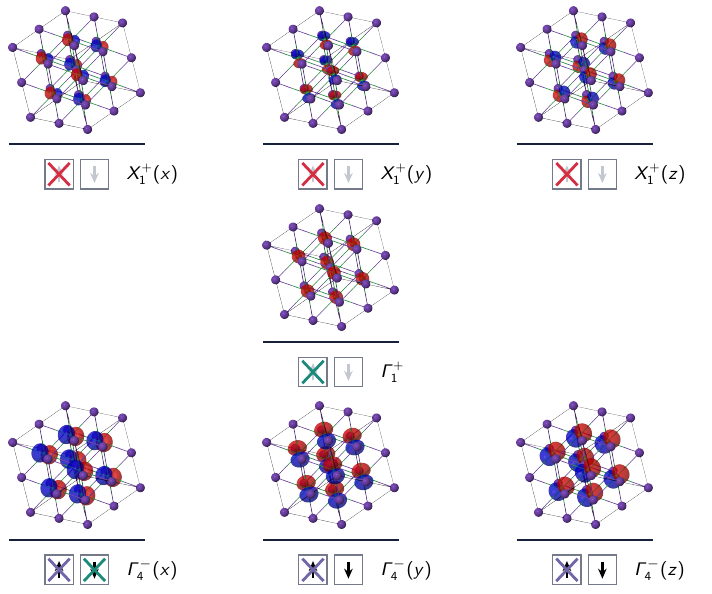}
  \caption[Caesium chloride (CsCl) active-space encoding]{%
    Caesium chloride (CsCl) CAS(6,7) active-space encoding.
    Spatial orbitals are shown as isosurfaces (red and blue denote opposite
    phases), and each adjacent box represents the occupation of one spin
    orbital, encoded as one qubit by the Jordan--Wigner mapping.
    Black (light-grey) arrows denote occupied (unoccupied) spin orbitals in the
    closed-shell KRHF reference determinant, represented in the folded
    $\Gamma$-supercell active basis.
    Coloured crosses mark the eight pivot spin-orbital occupations eliminated
    as redundant by the affine encoding.  Their colours indicate the
    symmetry-generator class assigned to each pivot in a valid elimination
    ordering: two spin-parity constraints (\legendsquare{spinparity}), three
    half-translation constraints (\legendsquare{halftranslation}), and three
    point-group reflections (\legendsquare{pointgroup}).  An individual
    row-reduced pivot constraint may combine generators from multiple classes.
    The occupations of redundant qubits are entirely determined by the retained
    occupations and the target symmetry sector, allowing the register in this
    case to be reduced from 14 to 6 qubits.
  }
  \label{fig:cscl_active_orbitals}
\end{figure*}

\subsection{Resource estimates and noiseless UCCSD-VQE}
\label{subsec:resources}

Figure~\ref{fig:resource_summary} reports the resource estimates and noiseless
UCCSD-VQE results for the ten-crystal benchmark suite: JW and SAE qubit,
Pauli-term, UCCSD-parameter, circuit-depth, and CNOT counts, together with
objective evaluations and converged VQE errors.
Every JW and SAE comparison uses the same symmetry-restored active-space
Hamiltonian $\op{H}_{\mathrm{sym}}$, active orbitals, and reference determinant;
only the encoding and the resulting ansatz representation change.
Appendix~\ref{app:generator-breakdown} gives the corresponding spin-parity,
translation, and other crystal-symmetry generators for each system.
For every system we compute the exact FCI energy in the physical fixed-
$(N_\uparrow,N_\downarrow)$ active-space sector and independently optimise
SAE-UCCSD.  Full JW-UCCSD is also optimised for nine systems; MgF$_2$ is omitted
because of its greater statevector cost, although its JW circuit resources and
SAE-VQE result remain in the figure.

\begin{figure*}[t]
  \centering
  \begin{minipage}{0.94\textwidth}
    \centering
    \makebox[\linewidth][c]{%
      \begin{minipage}[b]{0.17\linewidth}\centering
        \includegraphics[width=\linewidth]{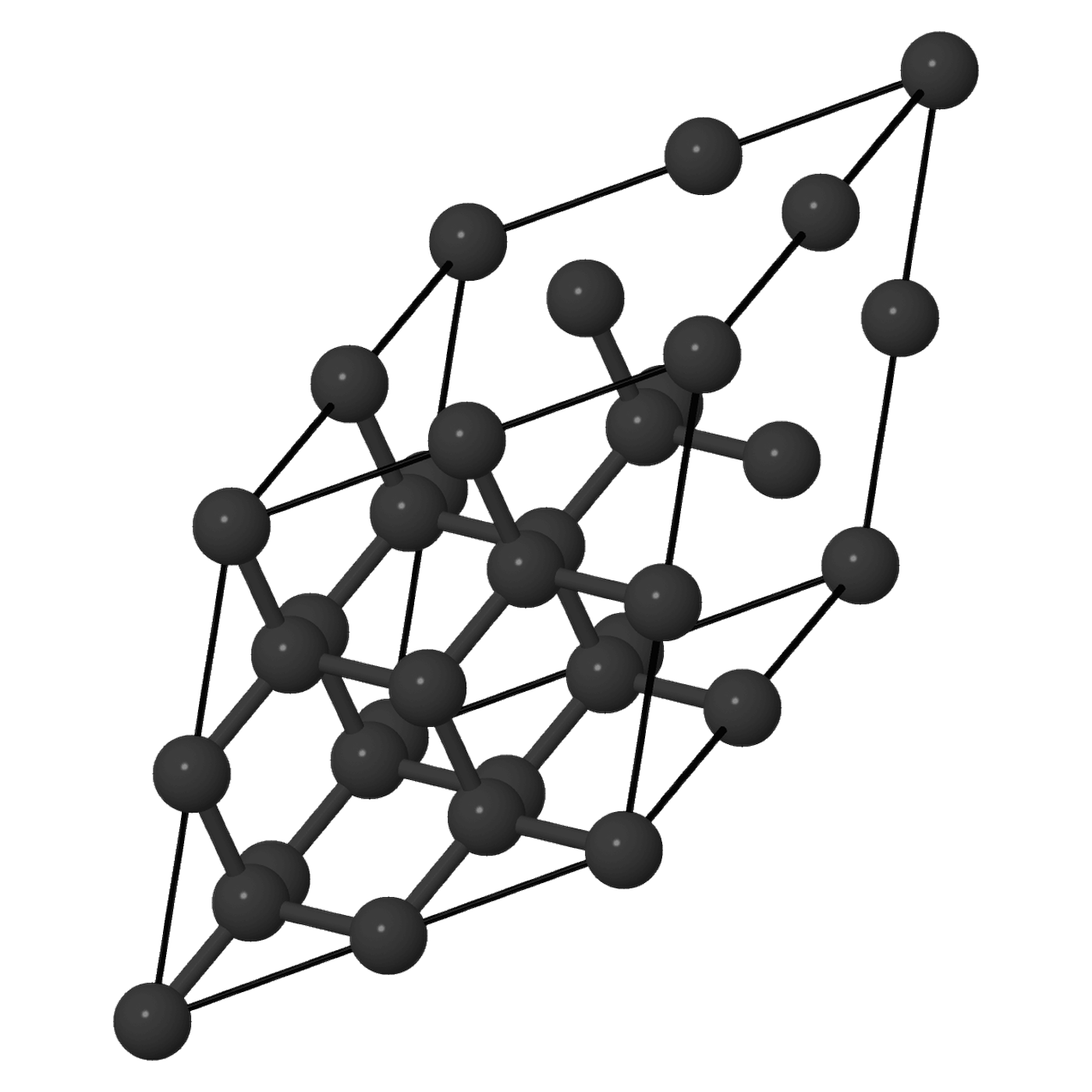}\\[-0.6ex]
        {\scriptsize\shortstack{\textbf{(a)} Diamond\\(C)\\
          {\color{black}$Fd\overline{3}m$}\\[-0.2ex]
          {\color{black}CAS(6,6)}}}
      \end{minipage}\hspace{0.013\linewidth}%
      \begin{minipage}[b]{0.17\linewidth}\centering
        \includegraphics[width=\linewidth]{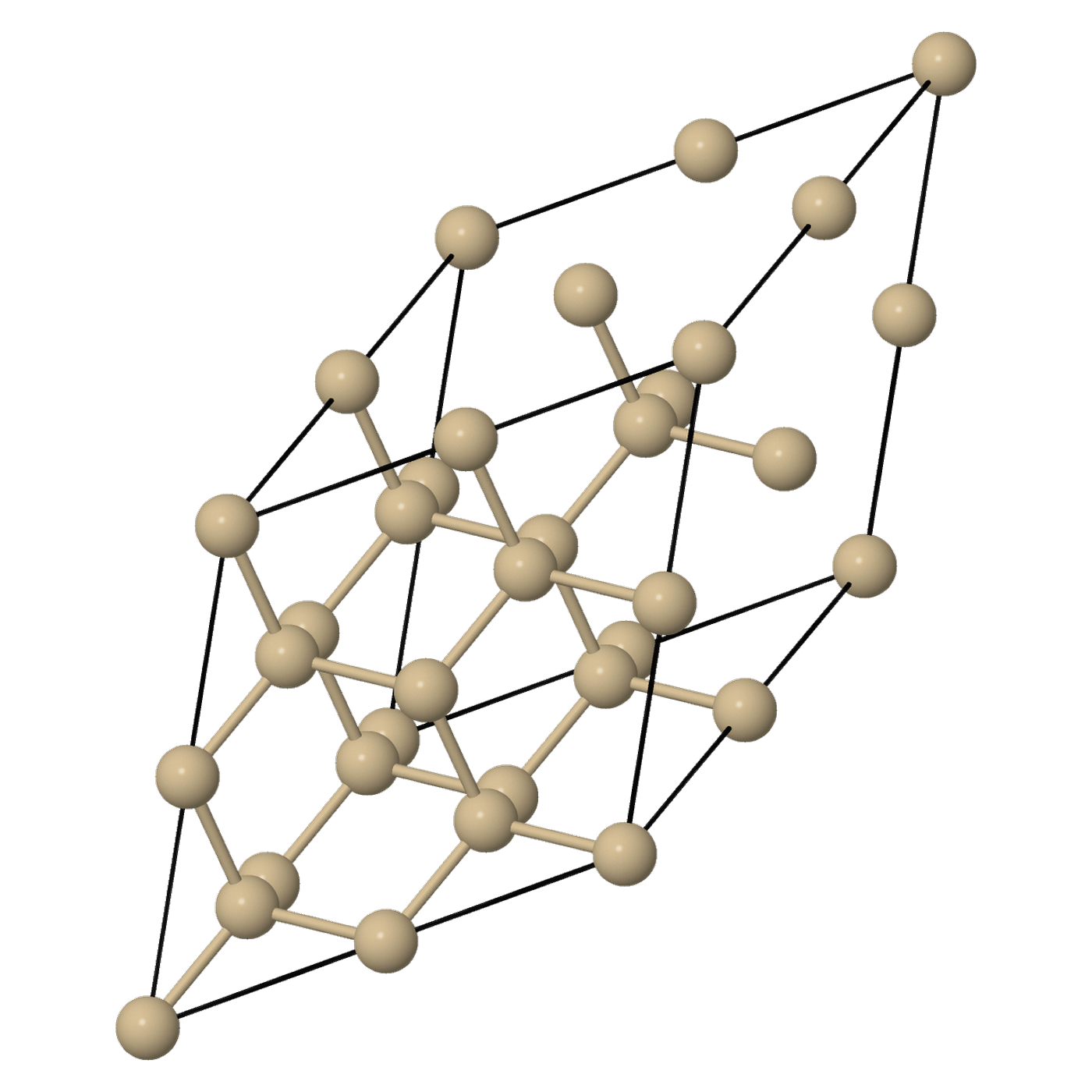}\\[-0.6ex]
        {\scriptsize\shortstack{\textbf{(b)} Silicon\\(Si)\\
          {\color{black}$Fd\overline{3}m$}\\[-0.2ex]
          {\color{black}CAS(6,7)}}}
      \end{minipage}\hspace{0.013\linewidth}%
      \begin{minipage}[b]{0.17\linewidth}\centering
        \includegraphics[width=\linewidth]{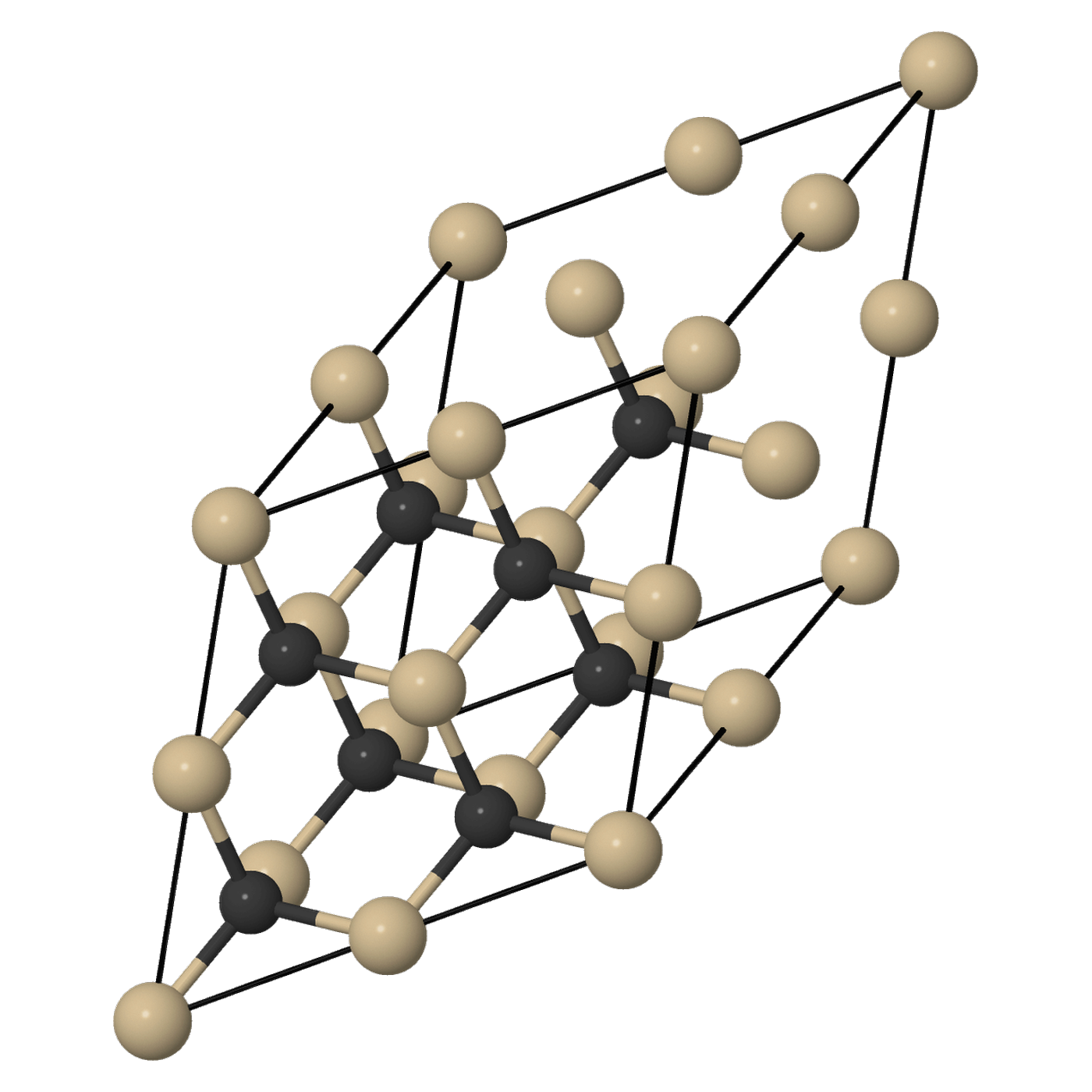}\\[-0.6ex]
        {\scriptsize\shortstack{\textbf{(c)} 3C silicon carbide\\(SiC)\\
          {\color{black}$F\overline{4}3m$}\\[-0.2ex]
          {\color{black}CAS(6,6)}}}
      \end{minipage}\hspace{0.013\linewidth}%
      \begin{minipage}[b]{0.17\linewidth}\centering
        \includegraphics[width=\linewidth]{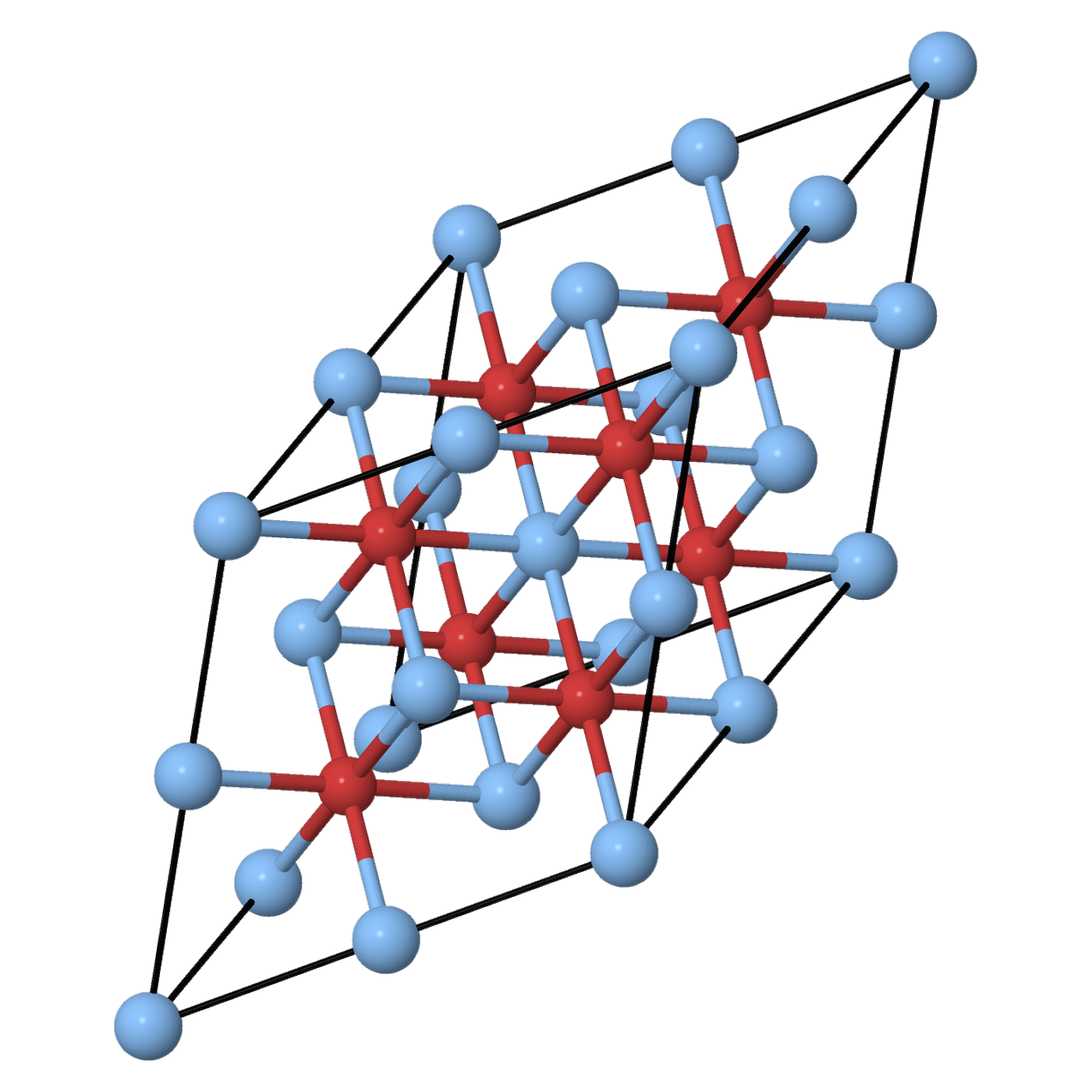}\\[-0.6ex]
        {\scriptsize\shortstack{\textbf{(d)} Magnesium oxide\\(MgO)\\
          {\color{black}$Fm\overline{3}m$}\\[-0.2ex]
          {\color{black}CAS(6,7)}}}
      \end{minipage}\hspace{0.013\linewidth}%
      \begin{minipage}[b]{0.17\linewidth}\centering
        \includegraphics[width=\linewidth]{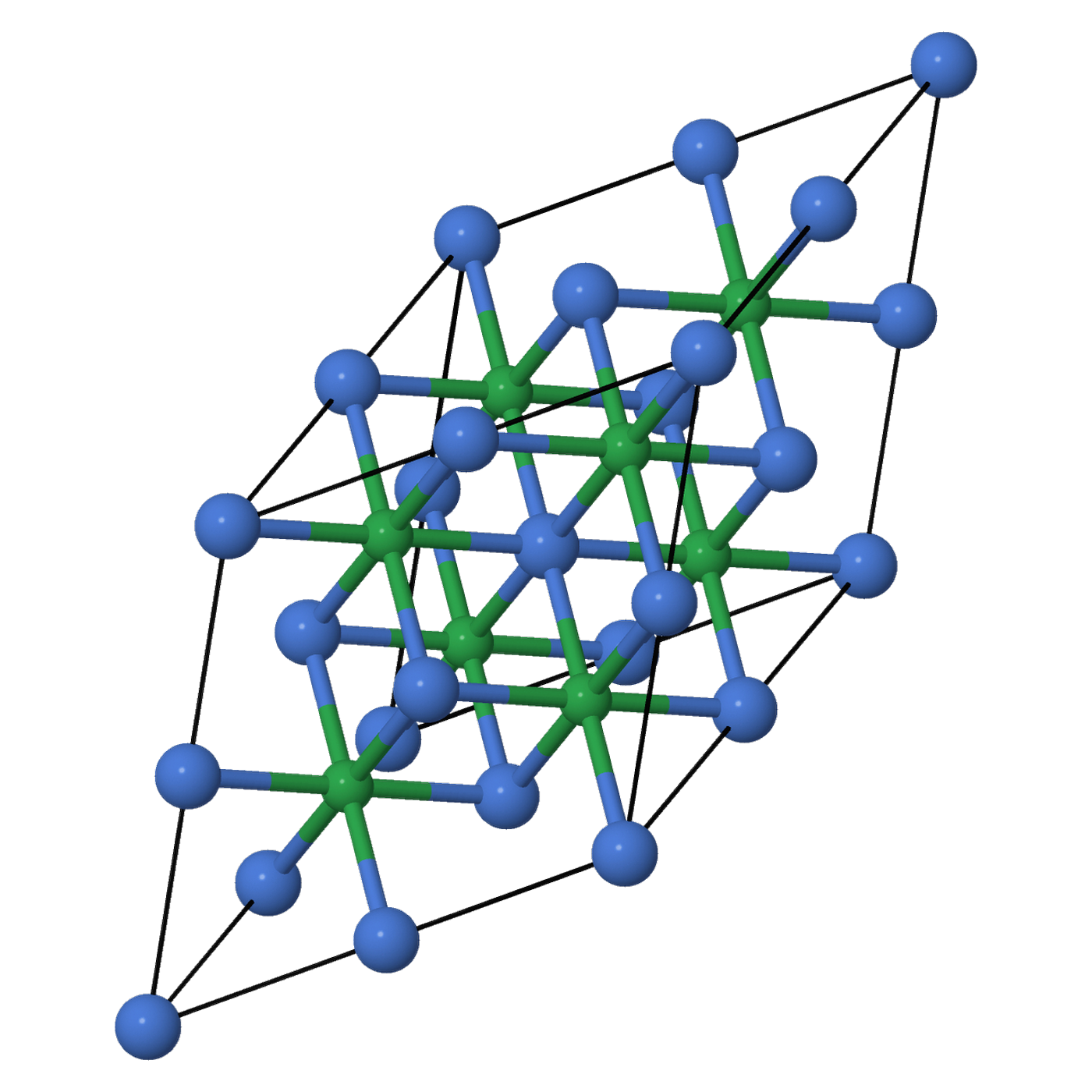}\\[-0.6ex]
        {\scriptsize\shortstack{\textbf{(e)} Sodium chloride\\(NaCl)\\
          {\color{black}$Fm\overline{3}m$}\\[-0.2ex]
          {\color{black}CAS(6,7)}}}
      \end{minipage}%
    }\par\vspace{1.8ex}
    \makebox[\linewidth][c]{%
      \begin{minipage}[b]{0.17\linewidth}\centering
        \includegraphics[width=\linewidth]{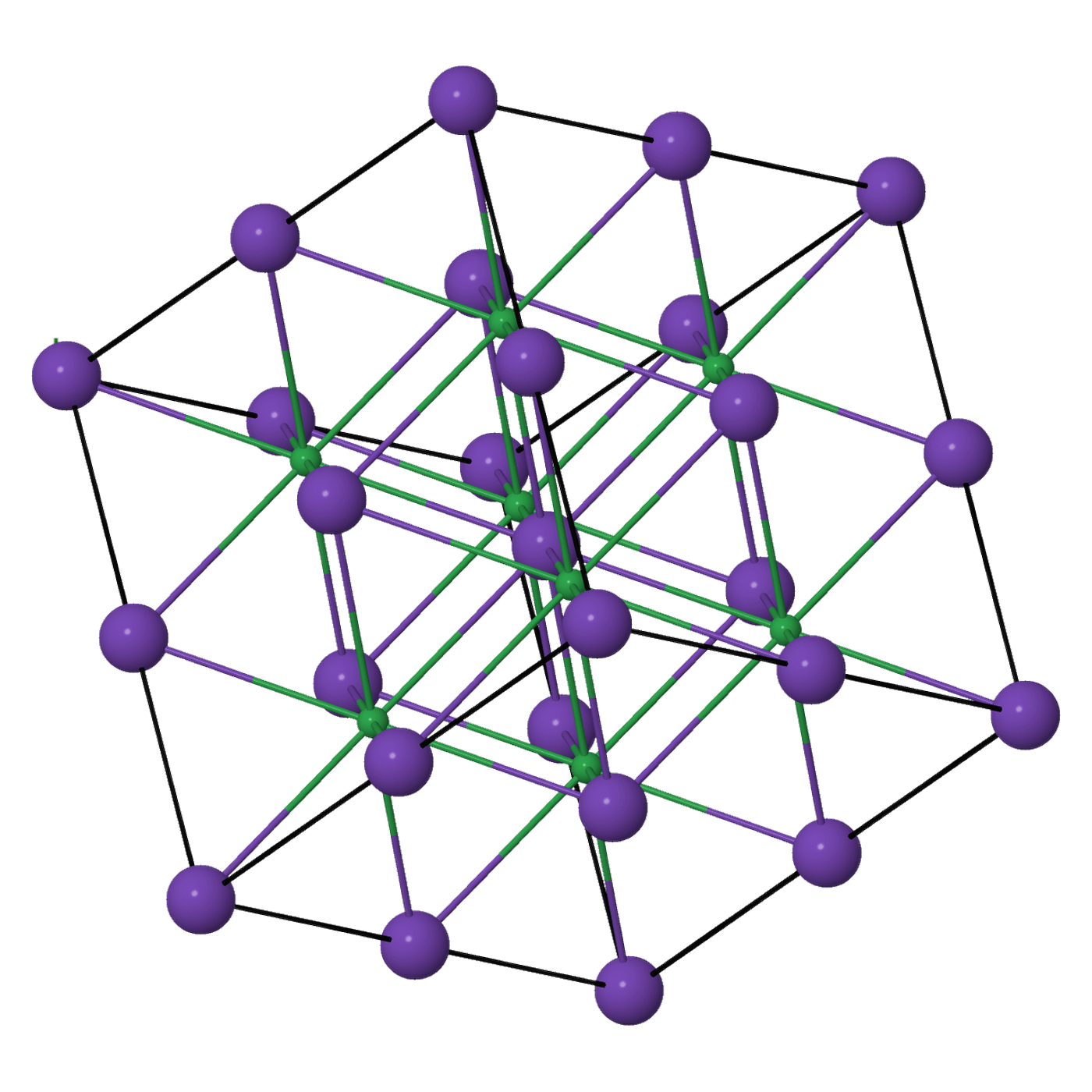}\\[-0.6ex]
        {\scriptsize\shortstack{\textbf{(f)} Caesium chloride\\(CsCl)\\
          {\color{black}$Pm\overline{3}m$}\\[-0.2ex]
          {\color{black}CAS(6,7)}}}
      \end{minipage}\hspace{0.013\linewidth}%
      \begin{minipage}[b]{0.17\linewidth}\centering
        \includegraphics[width=\linewidth]{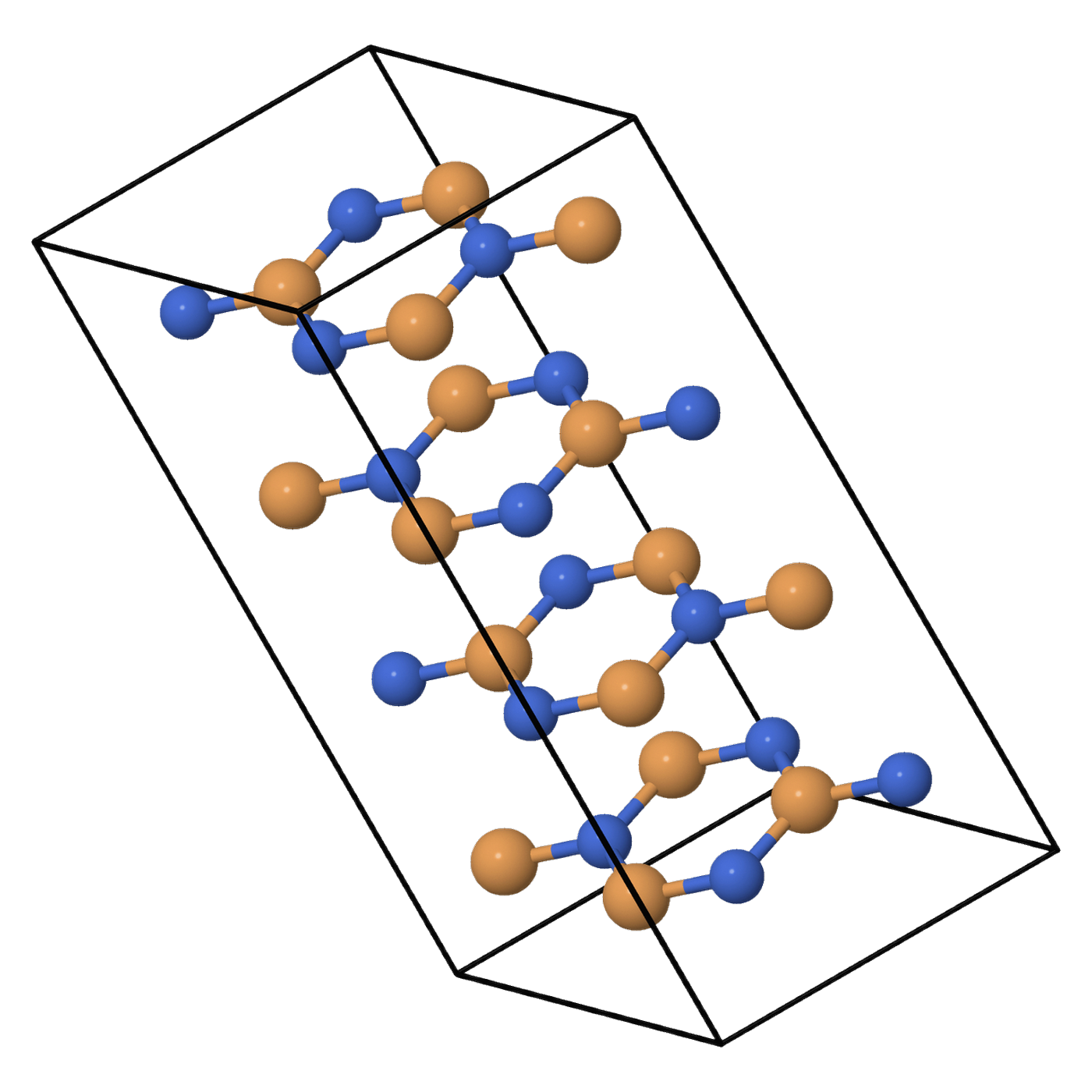}\\[-0.6ex]
        {\scriptsize\shortstack{\textbf{(g)} Hexagonal boron\\nitride (h-BN)\\
          {\color{black}$P6_3/mmc$}\\[-0.2ex]
          {\color{black}CAS(6,6)}}}
      \end{minipage}\hspace{0.013\linewidth}%
      \begin{minipage}[b]{0.17\linewidth}\centering
        \includegraphics[width=\linewidth]{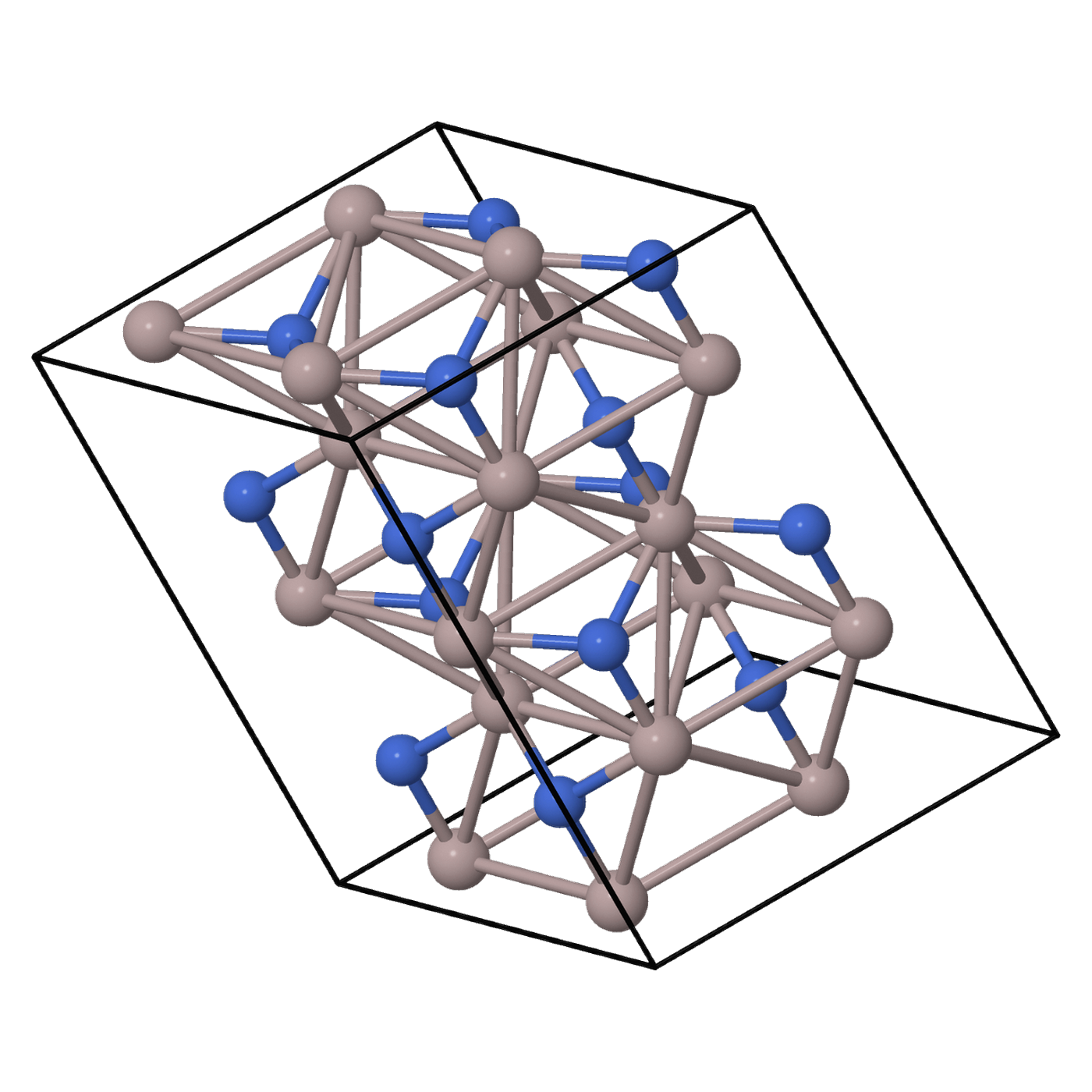}\\[-0.6ex]
        {\scriptsize\shortstack{\textbf{(h)} Aluminium nitride\\(AlN)\\
          {\color{black}$P6_3mc$}\\[-0.2ex]
          {\color{black}CAS(2,8)}}}
      \end{minipage}\hspace{0.013\linewidth}%
      \begin{minipage}[b]{0.17\linewidth}\centering
        \includegraphics[width=\linewidth]{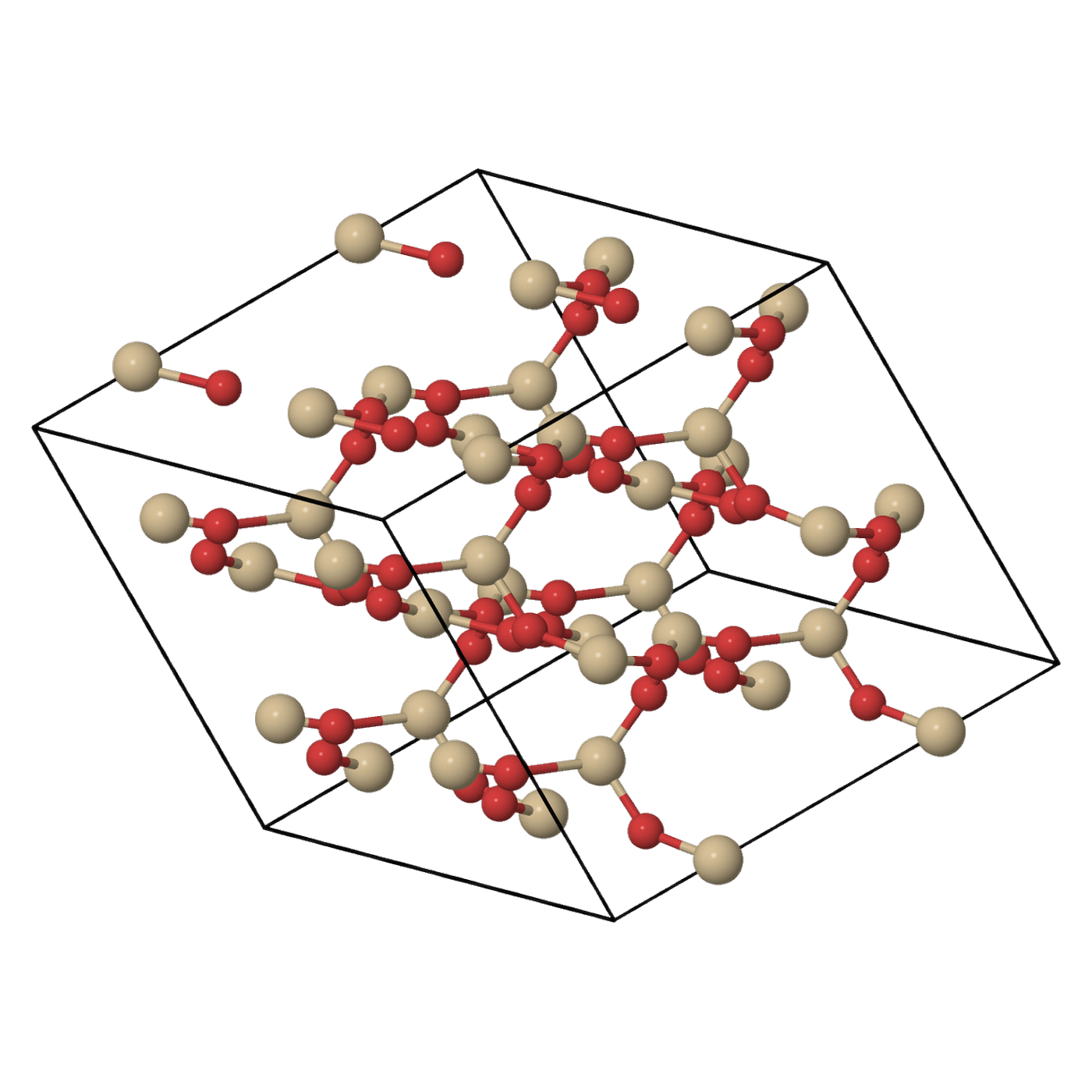}\\[-0.6ex]
        {\scriptsize\shortstack{\textbf{(i)} $\alpha$-quartz\\(SiO$_2$)\\
          {\color{black}$P3_121$}\\[-0.2ex]
          {\color{black}CAS(4,4)}}}
      \end{minipage}\hspace{0.013\linewidth}%
      \begin{minipage}[b]{0.17\linewidth}\centering
        \includegraphics[width=\linewidth]{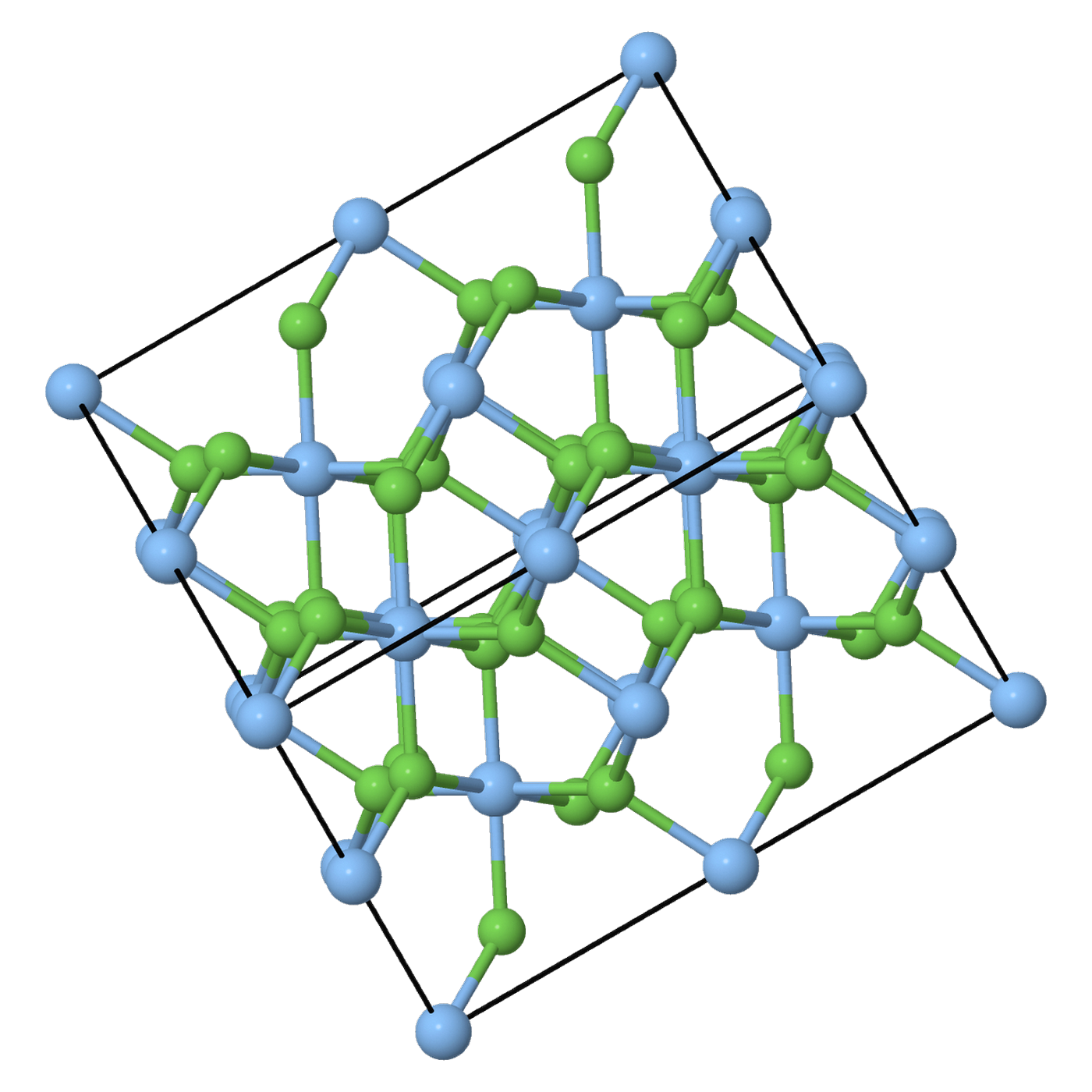}\\[-0.6ex]
        {\scriptsize\shortstack{\textbf{(j)} Magnesium fluoride\\(MgF$_2$)\\
          {\color{black}$P4_2/mnm$}\\[-0.2ex]
          {\color{black}CAS(6,8)}}}
      \end{minipage}%
    }
  \end{minipage}
  \par\vspace{1.2ex}
  \begin{minipage}{0.89\textwidth}
    \centering
    \includegraphics[width=\linewidth]{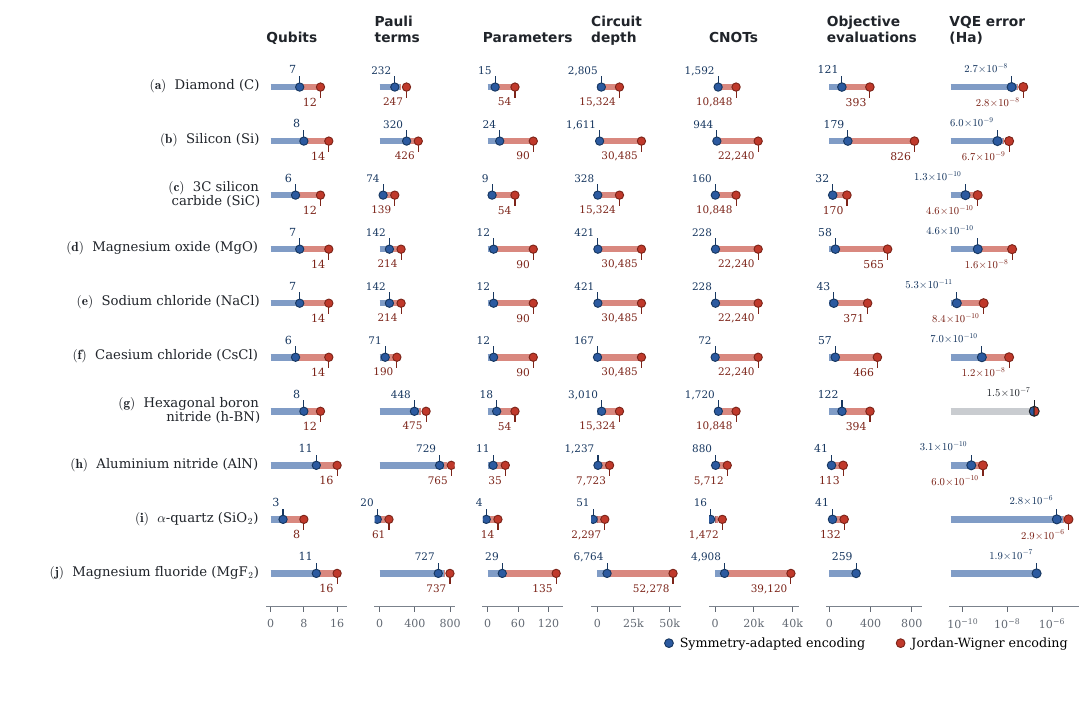}
  \end{minipage}
  \caption[Periodic benchmark systems and resource summary]{%
    The ten periodic benchmark systems, labelled \textbf{(a)}--\textbf{(j)} and
    shown as the $(2,2,2)$ supercells used in the calculations.  The set spans
    cubic, hexagonal, trigonal, and tetragonal structures.  Each supercell label
    gives the space group and CAS; the matching graph row reports encoding and
    circuit resources, optimisation effort, and noiseless UCCSD-VQE error for
    that active space.
    For each count metric, the blue segment ends at the SAE value and
    the complete blue--red bar ends at the corresponding Jordan--Wigner value,
    where available; labels give the exact values.
    The VQE-error column uses the same segmented-bar and endpoint convention;
    errors that coincide at the displayed precision merge into a neutral grey
    bar and a half-blue, half-red marker.
    VQE errors are measured relative to exact diagonalisation in the physical
    active-space sector.
    For MgF$_2$, the absent red bar and marker indicate that JW-VQE was omitted.
    The SAE reduction is exact within the target symmetry sector; where both
    optimisations were performed, it attains comparable
    errors with fewer objective evaluations.
  }
  \label{fig:benchmark_supercells}
  \label{fig:resource_summary}
\end{figure*}

\paragraph{Qubit reduction.}
The ten benchmark systems start from 8--16 JW qubits and reduce to 3--11 SAE
qubits.
Every system removes at least four qubits, and nine systems remove five or more.
CsCl gives the largest qubit reduction: its B2 structure and CAS(6,7) active
space admit eight independent Boolean generators, generating $\Z_2^8$ and
reducing the register from 14 to 6 qubits.
This saturates the maximum possible reduction in the present three-dimensional
Boolean-generator construction: two spin parities, up to three independent
translation characters, and up to three independent non-translation crystal
symmetries acting diagonally in the active space.
The three translation characters are what allow this periodic example to exceed
the five-generator maximum of molecular SAE.
Large reductions also occur in the non-cubic examples: AlN, $\alpha$-quartz,
and MgF$_2$ each remove five qubits.

\paragraph{Hamiltonian and UCCSD circuit resources.}
The Pauli-term reduction is system dependent.
It is large for CsCl and $\alpha$-quartz (63\% and 67\%, respectively), moderate
for NaCl, silicon, MgO, and 3C-SiC, and small for MgF$_2$, AlN, diamond, and h-BN.
The controlled comparison in Table~\ref{tab:jw-sf-resources} separates two
sources of circuit savings.  Symmetry filtering removes $67$--$87\%$ of the
UCCSD amplitudes and, on the full JW register, reduces the unoptimised CNOT count
by $72$--$97\%$.
Holding that indexed amplitude subset fixed, the SAE representation gives a
further $43$--$92\%$ CNOT reduction in eight systems.  The two exceptions are
AlN and MgF$_2$, for which this particular unoptimised product-formula
decomposition has 20\% and 10\% more CNOTs, respectively, after encoding.
Across the suite, the total JW-to-SAE reduction is $84$--$99.7\%$, combining
symmetry-filtering savings with those from reduced-register encoding.

For diamond, CsCl, and $\alpha$-quartz, the same two-stage saving persists after
level-3 compilation: under both all-to-all and nearest-neighbour-line
connectivity, SAE reduces CNOT counts by a further 26--91\% relative to JW-SF.
Periodic SAE therefore combines symmetry-forbidden-amplitude removal with, in
most cases, shorter representations of the surviving generators.

\paragraph{Noiseless VQE results and optimiser workload.}
Across all ten SAE-UCCSD-VQE calculations, the largest error relative to exact
FCI in the active-space sector is $2.84\times10^{-6}$~Ha, for
$\alpha$-quartz, well below the chemical-accuracy scale of
$1.6\times10^{-3}$~Ha.  These are ansatz-and-optimisation errors for fixed
active-space Hamiltonians, not estimates of convergence with basis set,
active-space size, pseudopotential, or $k$-point sampling.
For the nine systems with both full-JW and SAE UCCSD-VQE optimisations, both
optimisations converge well below chemical accuracy.  Their Hamiltonians contain the
same target-sector block, but their variational manifolds are not identical:
full JW retains every indexed spin-singlet UCCSD generator, whereas SAE
automatically omits generators whose target-sector image is zero.  Their
comparable final errors show empirically that this screening retains sufficient
variational flexibility for the present benchmarks.
Across those nine systems, finite-difference SLSQP uses 64--90\% fewer VQE
objective evaluations for SAE than for full JW
(Fig.~\ref{fig:resource_summary}).  This primarily reflects the smaller
symmetry-filtered parameter set and the associated reduction in finite-difference
evaluations per gradient, not the smaller qubit register alone.

\subsection{Benchmark validation}
\label{subsec:validation}

For each benchmark we form the unreduced JW matrix restricted to the complete
joint Boolean sector $Aa=c$ and compare its entire ordered spectrum with that
of the SAE Hamiltonian.  The two matrices have the common dimension
$2^{\nso-n_g}$, and the largest discrepancy over all levels and all systems is
$1.21\times10^{-11}$~Ha.  This directly tests isospectrality over the complete
selected Boolean sector.  The physical VQE reference is instead obtained by
exact diagonalisation after explicitly restricting the unreduced Hamiltonian
to fixed $(N_\uparrow,N_\downarrow)$.

The validation also reconstructs the folded KRHF determinant energy, checks the
GDF contractions through an independent assembly path, and confirms term by
term that every packed UCCSD generator has a definite symmetry character and
that character screening agrees with explicit SAE projection.  Detailed
spectral, character, and folded-integral results are given in
Appendix~\ref{app:reduction-validation}.

\FloatBarrier

% ============================================================
\section{Discussion}
\label{sec:discussion}
% ============================================================

Periodic SAE does not discard physical orbitals.  Instead, it removes occupation
coordinates whose values are fixed by the affine constraints $Aa=c$, representing
each eliminated coordinate as a function of the retained occupations and the
chosen sector.  The group average in Eq.~\eqref{eq:symmetry-restoration} removes
finite-precision couplings between the retained spatial-character sectors while
leaving every within-sector block unchanged; the affine projection then
represents the chosen block on fewer qubits.  The reduction therefore changes
the representation, not the symmetry-restored active-space energy problem,
consistent with the complete-spectrum agreement to $1.21\times10^{-11}$~Ha.

The distinctive periodic resource is the finite crystal translation group.
Translations whose active-space actions are non-trivial involutions can augment
the two spin parities and other crystal-symmetry generators, allowing periodic
SAE to exceed the molecular rank bound.  In the $(2,2,2)$ CsCl benchmark, the
three coordinate half translations contribute three of the eight eliminated
qubits.  On a larger mesh a half translation can instead act as $+I$ on the
selected active manifold, while a shorter translation becomes involutory only
after restriction.  The saving is therefore jointly controlled by the mesh,
crystal symmetry, and active-space representation.

The JW/JW-SF/SAE controls separate two effects that would otherwise be
conflated.  JW-to-JW-SF savings arise from omitting symmetry-forbidden UCCSD
amplitudes on the original register, whereas JW-SF-to-SAE differences arise
from representing the same surviving indexed generators on fewer qubits.  The
latter lowers the CNOT count in eight of the ten unoptimised circuits; AlN and
MgF$_2$ show that a smaller register does not guarantee a shorter circuit under
every fixed synthesis.  All three representative optimised compilations retain
an additional JW-SF-to-SAE saving.  The comparable JW and SAE VQE errors show
that screening preserves sufficient variational flexibility for these
benchmarks, but do not imply that the two ansatz manifolds are identical.  The
lower finite-difference objective count is primarily a consequence of the
smaller screened parameter set rather than register size.

The present workflow is restricted to spin-free, time-reversal-symmetric KRHF
calculations and real active-space involutions.  It treats positive diagonal
meshes by enumerating the finite translations and translated space-group
representatives, but does not include complex, magnetic, antiunitary, or
double-group representations.  Exactness refers to the selected Boolean sector of the finite,
symmetry-restored active-space Hamiltonian and does not establish convergence
with basis set, pseudopotential, active-space size, or $k$-point sampling.
Moreover, noiseless statevector VQE restricts the benchmarks to small active
spaces, and the circuit counts are logical resources; the representative
compilations use abstract connectivity and one transpiler seed rather than a
device-specific noise model.

% ============================================================
\section{Conclusion}
\label{sec:conclusion}
% ============================================================

We have developed a periodic SAE workflow that contracts primitive-cell
density-fitted $k$-point quantities directly into a real $\Gamma$-supercell
active space and incorporates commuting real active-space involutions derived
from crystal translations and space-group operations alongside spin parity.
Across ten crystals, periodic SAE removes
4--8 qubits; the CsCl CAS(6,7) example combines eight independent generators to
reduce 14 JW qubits to 6.  Complete-spectrum comparisons verify target-sector
isospectrality to numerical precision, and noiseless SAE-UCCSD reaches the
corresponding fixed-active-space FCI energies well within chemical accuracy
throughout the suite.  Relative to unscreened full-JW circuits, symmetry
screening removes 67--87\% of the indexed UCCSD parameters, while screening and
reduced-register encoding together lower the unoptimised logical CNOT counts by
84--99.7\%.

\begin{acknowledgments}
The author thanks Jonathan Tennyson, Alex Thom, and George Booth for helpful discussions.
This research is supported by an EPSRC Postdoctoral Prize Fellowship [EP/W524335/1] at UCL and an EPSRC Research Fellowship [EP/S021582/1] at the London Centre for Nanotechnology.
The author's PhD was supported by an EPSRC Industrial CASE studentship [EP/T517793/1].
\end{acknowledgments}

\section*{Code availability}

\texttt{QuantumSymmetry} v0.3.4 is available from
PyPI~\cite{QuantumSymmetrySoftware}; a
\href{https://github.com/dariopicozzi/quantumsymmetry/blob/master/docs/tutorials/08_periodic.ipynb}
{periodic-system tutorial} is available in the public source repository.

\clearpage
\onecolumngrid

\appendix

\section{Boolean-rank bounds}
\label{app:boolean-rank}

\begin{proposition}[Boolean-rank bound and its saturation]
\label{prop:boolean-rank}
Consider a non-relativistic, spin-conserving SAE constructed from the two spin
parities and mutually commuting real involutions in the active-space image of
the crystal symmetry group.  Assume that the active and frozen-core
one-particle spaces are invariant under the corresponding operations.  Let
$A$ be the resulting binary
spin-orbital character matrix, and let $e$ be the number of even entries among
the folded-mesh dimensions $(N_0,N_1,N_2)$.  Then
\begin{equation}
  \operatorname{rank}_{\GFtwo}A\leq 5 \quad \text{(molecular)},
  \qquad
  \operatorname{rank}_{\GFtwo}A\leq 5+e\leq 8
  \quad \text{(periodic)}.
  \label{eq:boolean-rank-bounds}
\end{equation}
The periodic upper bound is attained by the $(2,2,2)$ B2 CsCl CAS$(6,7)$
active space used in this work.
\end{proposition}

\begin{proof}
\emph{Upper bound.}
Let $\rho_{\mathcal A}$ be the real representation obtained after restricting
crystal operations to the invariant active space, and let $\mathcal E$ be the
elementary Abelian group generated by the retained commuting involutions in its
image.  Let $\mathcal E_T$ be the subgroup derived from the finite
Born--von K\'arm\'an translation group
$\mathcal T\simeq\mathbb Z_{N_0}\times\mathbb Z_{N_1}\times\mathbb Z_{N_2}$.
Although an active involution need not be represented by a global half
translation, every elementary Abelian 2-section of $\mathcal T$ has rank at
most $e$, the number of even cyclic factors.  Hence
$\dim_{\GFtwo}\mathcal E_T\leq e$.

The quotient $\mathcal E/\mathcal E_T$ is an elementary Abelian 2-section of a
three-dimensional crystallographic point group and has rank at most three.
Equivalently, no more than three independent non-translation characters can
remain after the translation image is factored out.  Adding the two spin
parities therefore gives
\begin{equation}
  \operatorname{rank}_{\GFtwo}A
  \leq 2+\dim_{\GFtwo}\mathcal E_T
       +\dim_{\GFtwo}(\mathcal E/\mathcal E_T)
  \leq5+e.
\end{equation}
A molecule has no translation subgroup and hence the bound five.  In three
dimensions $e\leq3$, giving the periodic bound eight.

\emph{Saturation by CsCl.}
For the $(2,2,2)$ CsCl supercell, $\mathbf A_i=2\mathbf a_i$.  The three half
translations $t_i=T_{\mathbf A_i/2}$ and the three coordinate reflections
$m_0=\sigma_{100}$, $m_1=\sigma_{010}$, and $m_2=\sigma_{001}$ are exact
involutions.  They commute modulo the supercell lattice: the reflections
commute with one another, while a reflection either fixes a half translation
or sends it to its inverse, which is the same translation modulo $\mathbf A_i$.

In the simultaneous-character active-orbital basis used in the calculation,
orbitals 62--68 transform as $\Gamma_4^-(x)$, $\Gamma_4^-(y)$, $\Gamma_4^-(z)$,
$\Gamma_1^+$, $X_1^+(x)$, $X_1^+(y)$, and $X_1^+(z)$.  For $i=x,y,z$, the
reflection normal to $i$ has character $-1$ only on $\Gamma_4^-(i)$, and the
half translation along $i$ has character $-1$ only on $X_1^+(i)$.  These six
distinct pivots prove that the spatial character matrix has rank six.

All six spatial generators have character $+1$ on $\Gamma_1^+$.  On its two
spin orbitals, however, $P_\uparrow$ and $P_\downarrow$ have binary characters
$(1,0)$ and $(0,1)$.  The two spin-parity rows are therefore independent of
one another and of the six spatial rows, so
$\operatorname{rank}_{\GFtwo}A_{\mathrm{CsCl}}=8$.  Since affine SAE removes
one qubit per independent row, the 14-spin-orbital register reduces to six
qubits.
\end{proof}

\section{Sector-preserving numerical symmetry restoration}
\label{app:symmetry-restoration}

The propositions below use the exact sign-restored representatives defined in
the Methods section; the reported removed-coefficient one-norm quantifies the
associated finite numerical regularisation.

Assume that the selected active one-particle space is invariant under the
retained finite spatial Boolean group $G_{\mathrm{sp}}\simeq\mathbb Z_2^r$.
Let $U_g$ be its
unitary representation on the active-space Fock space, and let $P_c$ project
onto the joint character sector $c$.  Only the spatial group enters this
restoration; spin parity is already exact because the Hamiltonian conserves
$N_\uparrow$ and $N_\downarrow$.

\begin{proposition}[Symmetry averaging preserves every sector block]
\label{prop:symmetry-restoration}
For any active-space operator $X$, the spatial group average satisfies
\begin{equation}
\begin{aligned}
  \mathcal R_G(X)
  &:=\frac{1}{|G_{\mathrm{sp}}|}\sum_{g\in G_{\mathrm{sp}}}
      U_gXU_g^\dagger
    =\sum_cP_cXP_c,\\
  P_c\mathcal R_G(X)P_c&=P_cXP_c.
\end{aligned}
  \label{eq:sector-block-preservation}
\end{equation}
Moreover, $\mathcal R_G$ is the Hilbert--Schmidt orthogonal projector onto the
operators commuting with $G_{\mathrm{sp}}$.  In a simultaneous character basis
of active spatial orbitals, it retains exactly the symmetry-allowed one- and
two-electron tensor elements.
\end{proposition}

\begin{proof}
Conjugating the group average by any $U_h$ merely permutes its summands, so its
image commutes with $G_{\mathrm{sp}}$.  Conversely, every operator that already
commutes with the group is fixed by the average.  Group closure therefore makes
$\mathcal R_G$ idempotent, while closure under inversion makes it self-adjoint
in the Hilbert--Schmidt inner product.  It is thus the orthogonal projector onto
the commutant.

Because the group is Boolean, its characters are real and
$U_g=\sum_c\chi_c(g)P_c$.  Their orthogonality relation is
\begin{equation}
  \frac{1}{|G_{\mathrm{sp}}|}\sum_{g\in G_{\mathrm{sp}}}
  \chi_c(g)\chi_d(g)=\delta_{cd}.
\end{equation}
Substitution into the group average gives
Eq.~\eqref{eq:sector-block-preservation}: averaging removes only inter-sector
couplings and leaves every sector block unchanged.

Finally, let $\eta_p(g)\in\{\pm1\}$ be the character of active spatial orbital
$p$, shared by its two spin orbitals.  The monomial
$a_{p\sigma}^\dagger a_{q\sigma}$ has character $\eta_p\eta_q$, while
$a_{p\sigma}^\dagger a_{q\tau}^\dagger a_{s\tau}a_{r\sigma}$ has character
$\eta_p\eta_q\eta_s\eta_r$.  A non-trivial Boolean character averages to zero
and the trivial character to one.  Thus the integral-level projection sets
precisely the symmetry-forbidden one- and two-electron tensor elements to zero.
\end{proof}

\section{Folded active-space integrals and Ewald convention}
\label{app:folded-integrals}

\subsection{Forward and inverse \texorpdfstring{$k$-to-$\Gamma$}{k-to-Gamma}
transformation}

For the integral transformations below, let $\mathcal K$ denote the complete
$N_k$-point mesh and resolve each final real supercell orbital into its
primitive-cell $k$-point components:
\begin{align}
  B^{\mathbf k}_{\mu p}
  &=\frac{1}{\sqrt{N_k}}\sum_{\mathbf R}e^{-i\mathbf k\cdot\mathbf R}
    C^{\mathrm{SC}}_{\mathbf R\mu,p},
  \label{eq:inverse-k-fold}\\
  \sum_{\mathbf k\in\mathcal K}
  B^{\mathbf k\dagger}S^{\mathbf k}B^{\mathbf k}&=I,
  \label{eq:fold-orthonormality}
\end{align}
where $S^{\mathbf k}$ is the primitive-cell AO overlap matrix.  On the complete
commensurate mesh, the forward and inverse Fourier matrices are conjugate
transposes.  Consequently, no additional factor of $1/N_k$ enters the
one-electron contractions below.  Because $W$ can mix different $k$-point
components, a symmetry-adapted supercell orbital need not have a definite
primitive-cell crystal momentum.

\subsection{Direct one- and two-electron active contractions}

Here $u,v,w,x$ label active spatial orbitals; spin is restored in
Eq.~\eqref{eq:ham} in the usual spin-diagonal form.  Before frozen-core and
Ewald contributions are added, the one-electron active block is
\begin{equation}
  h^{(0)}_{uv}=\sum_{\mathbf k\in\mathcal K}
  B_u^{\mathbf k\dagger}h^{\mathbf k}B_v^{\mathbf k},
  \label{eq:folded-h1}
\end{equation}
where $h^{\mathbf k}$ is the primitive-cell core Hamiltonian.  In chemists'
notation, the supercell active two-electron block is
\begin{align}
  (uv|wx)_{\mathrm{SC}}
  &=\frac{1}{N_k}
    \sum_{\mathbf k_i,\mathbf k_j,\mathbf k_\ell\in\mathcal K}
    \bigl(u^{\mathbf k_i}v^{\mathbf k_j}
    \big|w^{\mathbf k_\ell}x^{\mathbf k_m}\bigr)_{\mathrm{GDF}},
  \label{eq:folded-eri}\\[-0.3ex]
  \mathbf k_m&\equiv
  \mathbf k_i-\mathbf k_j+\mathbf k_\ell\pmod{\mathbf G}.
  \label{eq:k-conservation}
\end{align}
Here $u^{\mathbf k}$ denotes the AO coefficient column
$B^{\mathbf k}_{\cdot u}$, and each term is evaluated by a primitive-cell GDF
contraction from the AO basis to the selected supercell orbitals.  The factor
$1/N_k$ is applied once, after momentum
conservation has reduced the nominal four-$k$ sum to a three-$k$ sum.  With
this convention,
\begin{equation}
  \op{H}_2=\frac12\sum_{uvwx}\sum_{\sigma\tau}
  (uv|wx)_{\mathrm{SC}}
  a_{u\sigma}^\dagger a_{w\tau}^\dagger a_{x\tau}a_{v\sigma}.
  \label{eq:spatial-h2-convention}
\end{equation}
Equation~\eqref{eq:folded-eri} is accumulated one
momentum-conserving quartet at a time.  It therefore requires
$O(n_{\mathrm{act}}^4)$ output memory rather than the
$O(N_k^3n_{\mathrm{act}}^4)$ memory required to materialise the complete
seven-dimensional $k$-resolved tensor.  The time-reversal-complete sum is real;
the residual imaginary part is checked before the real tensor is retained.

\subsection{Frozen core and Ewald exchange correction}

The frozen and active spaces are selected as unions of complete folded-energy
blocks.  The closed-shell frozen-core projector is therefore block diagonal in
$\mathbf k$:
\begin{equation}
  D_{\mathrm c}^{\mathbf k}
  =2B_{\mathrm c}^{\mathbf k}B_{\mathrm c}^{\mathbf k\dagger},
  \qquad
  N_{\mathrm c}=\sum_{\mathbf k}
  \operatorname{Tr}(D_{\mathrm c}^{\mathbf k}S^{\mathbf k}).
  \label{eq:core-density}
\end{equation}
Let
\begin{equation}
  V_{\mathrm c}^{\mathbf k,(0)}
  =J^{\mathbf k}[\{D_{\mathrm c}\}]
  -\frac12K^{\mathbf k,(0)}[\{D_{\mathrm c}\}]
  \label{eq:core-potential}
\end{equation}
be its restricted-Hartree--Fock potential before the special Ewald exchange
term.  The effective active one-electron operator and frozen-core constant are
then
\begin{align}
  h_{uv}^{\mathrm{eff}}
  &=h_{uv}^{(0)}+\sum_{\mathbf k}
    B_u^{\mathbf k\dagger}V_{\mathrm c}^{\mathbf k,(0)}B_v^{\mathbf k}
    -\frac{\xi_M}{2}\delta_{uv},
  \label{eq:active-h1-ewald}\\
  E_{\mathrm{core}}
  &=E_{\mathrm{nuc}}^{\mathrm{SC}}
    +\sum_{\mathbf k}\operatorname{Tr}
      (D_{\mathrm c}^{\mathbf k}h^{\mathbf k})\notag\\
  &\quad+\frac12\sum_{\mathbf k}\operatorname{Tr}
      (D_{\mathrm c}^{\mathbf k}V_{\mathrm c}^{\mathbf k,(0)})
    -\frac{\xi_M}{2}N_{\mathrm c}.
  \label{eq:core-energy-ewald}
\end{align}

We use the \textsc{PySCF} finite-supercell probe-charge convention
\begin{equation}
  \xi_M\equiv-2E_{\mathrm{Ew}}^{(1)}(\mathcal S),
  \label{eq:madelung-definition}
\end{equation}
where $E_{\mathrm{Ew}}^{(1)}(\mathcal S)$ is the Ewald electrostatic energy of
one unit probe charge and its neutralising background in the Born--von K\'arm\'an
supercell $\mathcal S$, whose lattice vectors are
$\mathbf A_i=N_i\mathbf a_i$~\cite{Ewald1921,Fraser1996}.  Unlike a bare
reciprocal-space sum, this definition includes the real-space, reciprocal-space,
self, and background terms.

For a spin-summed closed-shell density, the corresponding PySCF exchange
correction is
\begin{equation}
  \Delta K^{\mathbf k}
  =\xi_M S^{\mathbf k}D^{\mathbf k}S^{\mathbf k}.
  \label{eq:ewald-k-shift}
\end{equation}
Using $D^{\mathbf k}S^{\mathbf k}D^{\mathbf k}=2D^{\mathbf k}$, its
supercell exchange-energy contribution is
\begin{align}
  \Delta E_x
  &=-\frac14\sum_{\mathbf k}\operatorname{Tr}
    (D^{\mathbf k}\Delta K^{\mathbf k})
   =-\frac{\xi_M}{2}N_e,
  \label{eq:ewald-energy-shift}\\
  \Delta\op{H}&=-\frac{\xi_M}{2}\op{N}.
  \label{eq:ewald-number-operator}
\end{align}
Here $N_e=\sum_{\mathbf k}\operatorname{Tr}(D^{\mathbf k}S^{\mathbf k})$
is the total electron number in the Born--von K\'arm\'an supercell.
Equations~\eqref{eq:active-h1-ewald} and
\eqref{eq:core-energy-ewald} are the active/frozen-core partition of
Eq.~\eqref{eq:ewald-number-operator}.  The Ewald-corrected
$K_{\mathrm c}^{\mathbf k}$ is obtained directly from the KRHF
object with \texttt{exxdiv="ewald"}; its trace already supplies the
$-\xi_MN_{\mathrm c}/2$ contribution, which is not added a second time.  The
same corrected core exchange can be used in the active contraction because
complete, mutually orthogonal folded blocks obey
$B_{\mathrm{act}}^{\mathbf k\dagger}S^{\mathbf k}D_{\mathrm c}^{\mathbf k}
S^{\mathbf k}B_{\mathrm{act}}^{\mathbf k}=0$ after summation over $\mathbf k$.
The active contribution is applied explicitly as $-\xi_M\delta_{uv}/2$.  Thus the
frozen-core and active contributions apply the Ewald correction exactly once.

\section{Benchmark definitions and target sectors}

\subsection{Benchmark input parameters}
\label{app:benchmark-inputs}

Table~\ref{tab:benchmark-inputs} records the primitive-cell input parameters used
for the benchmark calculations.
Lattice lengths are in \AA; internal parameters are fractional coordinates in
the corresponding primitive-cell construction.
The active-orbital windows are one-based indices in the folded and
symmetry-adapted supercell-orbital ordering used in the benchmark calculations.
Except where a row states otherwise, the electronic-structure inputs are the
\texttt{gth-szv} basis and \texttt{gth-pade} pseudopotential.

\begingroup
\renewcommand{\arraystretch}{1.22}
\small
\setlength{\tabcolsep}{3pt}
\begin{center}
\begin{minipage}{\columnwidth}
\centering
\refstepcounter{table}\label{tab:benchmark-inputs}
\parbox{\columnwidth}{\centering
  \textbf{TABLE~\thetable:} Periodic benchmark inputs used for the results
  shown in Fig.~\ref{fig:resource_summary}.  The abbreviation fcc denotes
  face-centred cubic.
}
\vspace{0.4ex}

\begin{tabular}{@{}p{0.29\columnwidth}p{0.67\columnwidth}@{}}
\toprule
System & Primitive-cell specification \\
\midrule
Diamond (C) & fcc, $a=3.567$; active orbital indices 30--35. \\
Silicon (Si) & fcc, $a=5.431$; active orbital indices 30--36. \\
3C silicon carbide (SiC) & fcc, $a=4.360$; active orbital indices 30--35. \\
Magnesium oxide (MgO) & fcc, $a=4.211$; active orbital indices 62--68. \\
Sodium chloride (NaCl) & fcc, $a=5.640$; active orbital indices 30--36; basis: Na/Cl \texttt{gth-szv};
  pseudo: Na \texttt{gth-pade-q1}, Cl \texttt{gth-pade}. \\
Caesium chloride (CsCl) & simple cubic B2, $a=4.123$; active orbital indices 62--68; basis:
  \texttt{gth-szv-molopt-sr}; pseudo: \texttt{gth-pade}. \\
Hexagonal boron nitride (h-BN) & hexagonal, $a=2.504$, $c=6.661$; active orbital indices 62--67. \\
Aluminium nitride (AlN) & hexagonal wurtzite, $a=3.110$, $c/a=1.601$, $u=3/8$; active orbital indices
  64--71. \\
$\alpha$-quartz (SiO$_2$) &
  hexagonal, $a=4.913$, $c=5.405$, $u_{\rm Si}=0.4697$,
  $(x,y,z)_{\rm O}=(0.4133,0.2667,0.1188)$; active orbital indices 191--194. \\
Magnesium fluoride (MgF$_2$) & rutile tetragonal, $a=4.621$, $c=3.052$, $u_{\rm F}=0.304$;
  active orbital indices 190--197. \\
\bottomrule
\end{tabular}
\end{minipage}
\end{center}
\endgroup

\subsection{Symmetry generators and target sectors}
\label{app:generator-breakdown}

Table~\ref{tab:generator-breakdown} lists the independent Boolean generators
used for each benchmark system shown in Fig.~\ref{fig:resource_summary}.
The two spin-parity generators are present in every system.
The symbols $\mathbf{a}_i$ and $\mathbf{A}_i=N_i\mathbf{a}_i$ denote the
primitive and supercell lattice vectors, respectively.  Thus
$T_{\mathbf{A}_i/2}=T_{(N_i/2)\mathbf{a}_i}$ is a half-supercell translation;
for the $(2,2,2)$ meshes used here it is a translation by one primitive vector.
The labels $\sigma_{hkl}$ denote point-group reflections whose normals are
parallel to $h\mathbf{b}_0+k\mathbf{b}_1+l\mathbf{b}_2$, where
$\mathbf{b}_0,\mathbf{b}_1,\mathbf{b}_2$ are primitive reciprocal vectors;
overbars denote negative indices, and $i$ denotes inversion.  Since every
benchmark in this table uses a $(2,2,2)$ mesh, we abbreviate
$T_{m_0m_1m_2}=T_{(m_0\mathbf A_0+m_1\mathbf A_1+m_2\mathbf A_2)/2}$.
The labels $C_{2,[uvw]}$ denote twofold rotations about axes parallel to
$u\mathbf a_0+v\mathbf a_1+w\mathbf a_2$; products such as
$T_{011}C_{2,[uvw]}$ are translated space-group operations.

\begingroup
\renewcommand{\arraystretch}{1.08}
\small
\setlength{\tabcolsep}{3pt}
\begin{center}
\begin{minipage}{\columnwidth}
\centering
\refstepcounter{table}\label{tab:generator-breakdown}
\parbox{\columnwidth}{\centering
  \textbf{TABLE~\thetable:} Independent Boolean generators retained for each
  benchmark system.  Each character of the target-sign string $\mathbf s$ is
  listed in the same order as the generators across the three preceding columns
  and determines $c_j=(1-s_j)/2$.  The final column equals the JW-to-SAE
  qubit-count reduction shown in Fig.~\ref{fig:resource_summary}.
}
\vspace{0.6ex}

\begin{tabular}{@{}p{0.19\columnwidth}p{0.13\columnwidth}p{0.23\columnwidth}p{0.21\columnwidth}p{0.13\columnwidth}c@{}}
\toprule
System & Spin parities & Translations & Other crystal operations & $\mathbf s$ & $n_g$ \\
\midrule
Diamond (C) & $P_\uparrow$, $P_\downarrow$ &
  -- &
  $\sigma_{\bar{1}01}$, $\sigma_{010}$, $i$ & \texttt{--+++} & 5 \\
Silicon (Si) & $P_\uparrow$, $P_\downarrow$ &
  $T_{\mathbf{A}_0/2}$, $T_{\mathbf{A}_1/2}$, $T_{\mathbf{A}_2/2}$ &
  $i$ & \texttt{--++++} & 6 \\
3C silicon carbide (SiC) & $P_\uparrow$, $P_\downarrow$ &
  $T_{\mathbf{A}_2/2}$, $T_{\mathbf{A}_1/2}$ &
  $T_{011}C_{2,[\bar{1}11]}$, $T_{011}C_{2,[1\bar{1}1]}$ & \texttt{--++++} & 6 \\
Magnesium oxide (MgO) & $P_\uparrow$, $P_\downarrow$ &
  $T_{\mathbf{A}_2/2}$, $T_{\mathbf{A}_1/2}$ &
  $\sigma_{\bar{1}11}$, $i$, $\sigma_{1\bar{1}1}$ & \texttt{--+++++} & 7 \\
Sodium chloride (NaCl) & $P_\uparrow$, $P_\downarrow$ &
  $T_{\mathbf{A}_2/2}$, $T_{\mathbf{A}_1/2}$ &
  $\sigma_{\bar{1}11}$, $i$, $\sigma_{1\bar{1}1}$ & \texttt{--+++++} & 7 \\
Caesium chloride (CsCl) & $P_\uparrow$, $P_\downarrow$ &
  $T_{\mathbf{A}_0/2}$, $T_{\mathbf{A}_1/2}$, $T_{\mathbf{A}_2/2}$ &
  $\sigma_{100}$, $\sigma_{010}$, $\sigma_{001}$ & \texttt{--++++++} & 8 \\
Hexagonal boron nitride (h-BN) & $P_\uparrow$, $P_\downarrow$ &
  $T_{\mathbf{A}_0/2}$, $T_{\mathbf{A}_1/2}$ &
  -- & \texttt{--++} & 4 \\
Aluminium nitride (AlN) & $P_\uparrow$, $P_\downarrow$ &
  $T_{\mathbf{A}_0/2}$, $T_{\mathbf{A}_1/2}$, $T_{\mathbf{A}_2/2}$ &
  -- & \texttt{--+++} & 5 \\
$\alpha$-quartz (SiO$_2$) & $P_\uparrow$, $P_\downarrow$ &
  $T_{\mathbf{A}_0/2}$, $T_{\mathbf{A}_1/2}$, $T_{\mathbf{A}_2/2}$ &
  -- & \texttt{+++++} & 5 \\
Magnesium fluoride (MgF$_2$) & $P_\uparrow$, $P_\downarrow$ &
  $T_{\mathbf{A}_0/2}$, $T_{\mathbf{A}_1/2}$ &
  $i$ & \texttt{--+++} & 5 \\
\bottomrule
\end{tabular}
\end{minipage}
\end{center}
\endgroup

\section{Symmetry-filtered circuit resources}
\label{app:jw-sf-resources}

JW-SF applies SAE excitation screening on the unreduced JW Hamiltonian and
register, so its qubit and Pauli-term counts equal those of JW.

\begingroup
\renewcommand{\arraystretch}{1.08}
\small
\setlength{\tabcolsep}{4pt}
\begin{center}
\begin{minipage}{\columnwidth}
\centering
\refstepcounter{table}\label{tab:jw-sf-resources}
\parbox{\columnwidth}{\centering
  \textbf{TABLE~\thetable:} Unoptimised singlet-UCCSD circuit resources.
  Depth and CNOT counts are obtained after three successive Qiskit
  decompositions, before transpiler optimisation or connectivity constraints.
}
\vspace{0.5ex}

\begin{tabular}{@{}lrrrrrrrrr@{}}
\toprule
& \multicolumn{3}{c}{UCCSD parameters}
& \multicolumn{3}{c}{Circuit depth}
& \multicolumn{3}{c}{CNOT count} \\
\cmidrule(lr){2-4}\cmidrule(lr){5-7}\cmidrule(lr){8-10}
System
& JW & JW-SF & SAE
& JW & JW-SF & SAE
& JW & JW-SF & SAE \\
\midrule
Diamond (C)
& 54 & 15 & 15
& 15{,}324 & 3{,}751 & 2{,}805
& 10{,}848 & 2{,}864 & 1{,}592 \\
Silicon (Si)
& 90 & 24 & 24
& 30{,}485 & 3{,}351 & 1{,}611
& 22{,}240 & 2{,}368 & 944 \\
3C silicon carbide (SiC)
& 54 & 9 & 9
& 15{,}324 & 468 & 328
& 10{,}848 & 432 & 160 \\
Magnesium oxide (MgO)
& 90 & 12 & 12
& 30{,}485 & 560 & 421
& 22{,}240 & 576 & 228 \\
Sodium chloride (NaCl)
& 90 & 12 & 12
& 30{,}485 & 560 & 421
& 22{,}240 & 576 & 228 \\
Caesium chloride (CsCl)
& 90 & 12 & 12
& 30{,}485 & 560 & 167
& 22{,}240 & 576 & 72 \\
Hexagonal boron nitride (h-BN)
& 54 & 18 & 18
& 15{,}324 & 3{,}943 & 3{,}010
& 10{,}848 & 3{,}008 & 1{,}720 \\
Aluminium nitride (AlN)
& 35 & 11 & 11
& 7{,}723 & 1{,}311 & 1{,}237
& 5{,}712 & 736 & 880 \\
$\alpha$-quartz (SiO$_2$)
& 14 & 4 & 4
& 2{,}297 & 282 & 51
& 1{,}472 & 192 & 16 \\
Magnesium fluoride (MgF$_2$)
& 135 & 29 & 29
& 52{,}278 & 6{,}027 & 6{,}764
& 39{,}120 & 4{,}448 & 4{,}908 \\
\bottomrule
\end{tabular}
\end{minipage}
\end{center}
\endgroup

\section{Numerical validation}
\label{app:reduction-validation}

\subsection{Complete target-sector spectra}

Let $H_{\mathrm{JW},c}$ be the unreduced JW matrix restricted to computational
basis states satisfying the complete joint constraint $Aa=c$, and let
$H_{\mathrm{SAE}}$ be the reduced matrix.  Both are constructed from the same
symmetry-restored active-space Hamiltonian.
Both matrices have dimension
$d_c=2^{\nso-\operatorname{rank}_{\GFtwo}A}$.  After ordering their complete
eigenvalue lists, we report
\begin{equation}
  \Delta_{\mathrm{spec}}
  =\max_{0\leq m<d_c}
  \left|E_m(H_{\mathrm{JW},c})-E_m(H_{\mathrm{SAE}})\right|.
  \label{eq:spectrum-validation-metric}
\end{equation}
The folded-reference closure $\Delta_{\mathrm{HF}}$ compares the expectation
value of the active-space
Hamiltonian on the folded KRHF determinant with $N_kE_{\mathrm{KRHF}}$.
The largest complete-spectrum and HF-closure residuals are respectively
$1.21\times10^{-11}$ and $5.0\times10^{-10}$~Ha.  Before diagonalisation, the
largest Hermiticity
residuals of the restricted JW and SAE matrices are both
$6.9\times10^{-18}$~Ha.

\begingroup
\renewcommand{\arraystretch}{1.12}
\small
\setlength{\tabcolsep}{7pt}
\begin{center}
\refstepcounter{table}\label{tab:reduction-validation}
\parbox{\columnwidth}{\centering
  \textbf{TABLE~\thetable:} Complete-spectrum numerical validation of the periodic active-space
  construction.  $d_c$ is the common dimension of the complete joint Boolean
  sector and the SAE register; $\Delta_{\mathrm{spec}}$ is defined in
  Eq.~\eqref{eq:spectrum-validation-metric}, and $\Delta_{\mathrm{HF}}$ is the
  folded-reference closure.  All residuals are in Hartree.
}
\vspace{0.4ex}

\begin{tabular}{lccrcc}
\toprule
System & CAS & Qubits (JW/SAE) & $d_c$ & $\Delta_{\mathrm{spec}}$ &
  $\Delta_{\mathrm{HF}}$ \\
\midrule
Diamond (C) & (6,6) & 12/7 & 128 & $2.27\times10^{-13}$ & $5.00\times10^{-10}$ \\
Silicon (Si) & (6,7) & 14/8 & 256 & $1.14\times10^{-13}$ & $1.37\times10^{-11}$ \\
3C silicon carbide (SiC) & (6,6) & 12/6 & 64 & $1.28\times10^{-13}$ & $6.51\times10^{-11}$ \\
Magnesium oxide (MgO) & (6,7) & 14/7 & 128 & $1.14\times10^{-12}$ & $5.68\times10^{-13}$ \\
Sodium chloride (NaCl) & (6,7) & 14/7 & 128 & $3.55\times10^{-13}$ & $1.42\times10^{-14}$ \\
Caesium chloride (CsCl) & (6,7) & 14/6 & 64 & $4.55\times10^{-13}$ & $2.52\times10^{-10}$ \\
Hexagonal boron nitride (h-BN) & (6,6) & 12/8 & 256 & $4.55\times10^{-13}$ & $3.87\times10^{-12}$ \\
Aluminium nitride (AlN) & (2,8) & 16/11 & 2048 & $1.14\times10^{-12}$ & $5.83\times10^{-12}$ \\
$\alpha$-quartz (SiO$_2$) & (4,4) & 8/3 & 8 & $7.96\times10^{-13}$ & $9.09\times10^{-13}$ \\
Magnesium fluoride (MgF$_2$) & (6,8) & 16/11 & 2048 & $1.21\times10^{-11}$ & $1.30\times10^{-11}$ \\
\bottomrule
\end{tabular}
\end{center}
\endgroup

As a separate check of Eq.~\eqref{eq:excitation-screening}, all
706 packed singlet-UCCSD generators in the ten systems were assigned characters
term by term before projection.  As required by the packed singlet
construction, all monomials associated with a given amplitude have the same
character.  Of the 706 generators, 146 preserve every selected character and
560 change at least one character and project to zero; the character test and
explicit SAE projection agree in every case.

\subsection{GDF active-block assembly}

Table~\ref{tab:gdf-validation} checks the streaming contraction in
Eq.~\eqref{eq:folded-eri} for diamond and CsCl.  The active
electron-repulsion-integral (ERI) tensor is
recomputed with PySCF's vectorised \texttt{GDF.ao2mo\_7d} route and its complete
three-$k$ tensor is summed with the same single $1/N_k$ normalisation.  The two
routes use the same primitive-cell GDF data and folded active coefficients, so
this comparison directly tests momentum-quartet assembly and normalisation; it
is not an independent estimate of the GDF approximation error.  The separate
folded-HF closure jointly checks the one-body, frozen-core, two-body, and Ewald
bookkeeping on the reference determinant.

\begingroup
\renewcommand{\arraystretch}{1.12}
\small
\setlength{\tabcolsep}{7pt}
\begin{center}
\refstepcounter{table}\label{tab:gdf-validation}
\parbox{\columnwidth}{\centering
  \textbf{TABLE~\thetable:} GDF assembly checks on the $(2,2,2)$
  mesh.  ERI residuals compare the streamed active tensor with
  \texttt{GDF.ao2mo\_7d}; RMS denotes root-mean-square, and ``0'' denotes
  bitwise equality of the returned
  double-precision tensors.  $\Delta_{\mathrm{HF}}^{\mathrm{raw}}$ uses the
  physical active Hamiltonian before symmetry restoration.  Values are in
  Hartree.
}
\vspace{0.4ex}

\begin{tabular}{lccccc}
\toprule
System & CAS & $\max|\Delta(uv|wx)|$ & RMS ERI residual &
  $\max|\operatorname{Im}(uv|wx)|$ & $\Delta_{\mathrm{HF}}^{\mathrm{raw}}$ \\
\midrule
Diamond (C) & (6,6) & 0 & 0 & $1.76\times10^{-34}$ & $5.60\times10^{-11}$ \\
Caesium chloride (CsCl) & (6,7) & $3.35\times10^{-13}$ & $1.68\times10^{-14}$ &
  $1.30\times10^{-16}$ & $5.70\times10^{-10}$ \\
\bottomrule
\end{tabular}
\end{center}
\endgroup

\subsection{General-mesh check with non-self-inverse \texorpdfstring{$k$}{k} points}

The ten resource benchmarks use $(2,2,2)$ meshes, whose $k$ points are all
self-inverse modulo a reciprocal-lattice vector.  To exercise the formation of
real supercell orbitals from complex $\mathbf k,-\mathbf k$ pairs and the
explicit-mesh phase convention, we also perform an end-to-end CsCl calculation
on a $(4,2,2)$ mesh.  Eight of its 16 $k$ points are non-self-inverse.  We
retain CAS(6,7) and the $3+1+3$ block composition of the $(2,2,2)$ CsCl
benchmark: the highest occupied triplet (folded orbitals 126--128), the lowest
unoccupied singlet (orbital 129), and the next unoccupied triplet (orbitals
132--134).  This is a complete union of the
energy blocks used for symmetry adaptation.  The selection is non-contiguous
because the intervening unoccupied doublet (orbitals 130--131) is deliberately
excluded.

The active character matrix has rank eight.  Its independent physical
generators are $P_\uparrow^{-}$, $P_\downarrow^{-}$,
$T_{\mathbf A_1/2}^{+}$, $T_{\mathbf A_2/2}^{+}$,
$T_{\mathbf a_0}^{+}$, $\sigma_{100}^{+}$, $\sigma_{010}^{+}$, and
$\sigma_{001}^{+}$, where the superscripts denote target eigenvalues.  All
three coordinate half translations remain exact operations.  Along axis 0,
however, the half-supercell shift is $T_{\mathbf A_0/2}=T_{2\mathbf a_0}$ and
acts as $+I$ on this active space, so it adds no constraint.  Its shorter square
root $T_{\mathbf a_0}$ has order four on the full $(4,2,2)$ supercell but its
restricted active-space action squares to $+I$ and has non-trivial $\pm1$
characters.  The active-space procedure therefore recovers the third
translation constraint and reduces 14 to 6 qubits, just as for the $(2,2,2)$
CAS(6,7) benchmark.
The folded orbitals are orthonormal to $1.8\times10^{-12}$, and the one-electron
block obtained from the primitive-$k$ GDF components agrees with the same block
formed by the explicit Fourier phase matrix to $2.2\times10^{-14}$~Ha.

The complete 64-level target-sector spectrum agrees with the SAE spectrum to
$1.4\times10^{-12}$~Ha.  The folded-KRHF reference closure residual is
$8.0\times10^{-13}$~Ha.  The
streaming and \texttt{GDF.ao2mo\_7d} active ERI tensors
agree to $5.8\times10^{-14}$~Ha, with a maximum residual imaginary component
of $2.2\times10^{-16}$~Ha.  Of 90 packed UCCSD generators, 12 preserve the
target sector and 78 are automatically screened, with no indefinite character
and no projection mismatch.
\clearpage
\twocolumngrid
\addtolength{\textheight}{\baselineskip}
\setlength{\bibsep}{0pt}

\bibliography{references}

\end{document}